\documentclass[11pt, letterpaper]{article}

\usepackage{paratype}
\usepackage{eulervm}
\usepackage{microtype}

\usepackage{complexity}

\usepackage[hidelinks]{hyperref}

\usepackage{amsmath}
\usepackage{amsfonts}
\usepackage{amssymb}
\usepackage{cleveref}

\usepackage{fullpage}
\usepackage{titling}

\usepackage{amsthm}
\usepackage{thmtools}
\usepackage[disable]{todonotes}
\usepackage{thm-restate}
\newtheorem{theorem}{Theorem}
\newtheorem{openproblem}{Open Problem}
\newtheorem{definition}{Definition}

\newtheorem{lemma}{Lemma}
\newtheorem{corollary}{Corollary}

\newcommand{\problem}[1]{\textsl{#1}}

\newlang{\OV}{OV}
\newlang{\OVC}{OVC}
\newlang{\UNSAT}{UNSAT}
\newlang{\linSAT}{linSAT}
\newlang{\MAX}{MAX}
\newlang{\ETHR}{ETHR}
\newlang{\THR}{THR}
\newlang{\ETH}{ETH}
\newlang{\NETH}{NETH}
\newlang{\SETH}{SETH}
\newlang{\NSETH}{NSETH}
\newlang{\TSP}{TSP}

\DeclareMathOperator{\rank}{rank}
\DeclareMathOperator{\size}{size}

\title{Conditional Complexity Hardness: Monotone Circuit Size, Matrix Rigidity, and Tensor Rank}

\author{
   Nikolai Chukhin
   \thanks{Neapolis University Pafos and JetBrains Research. Email: \url{buyolitsez1951@gmail.com}}
   \and
   Alexander~S. Kulikov
   \thanks{JetBrains Research. Email: \url{alexander.s.kulikov@gmail.com}}
   \and
   Ivan Mihajlin
   \thanks{JetBrains Research. Email: \url{ivmihajlin@gmail.com}}
   \and
   Arina Smirnova
   \thanks{Neapolis University Pafos. Email: \url{arina.00.ds@gmail.com}}
}
\date{}

\begin{document}
    \sloppy

	\maketitle

	\begin{abstract}
		Proving complexity lower bounds remains a~challenging task: currently,
		we~only know how to~prove conditional uniform (algorithm) lower bounds
		and nonuniform (circuit) lower bounds in~restricted circuit models.
		About a~decade ago, Williams (STOC 2010) showed how to~derive nonuniform lower bounds
		from uniform upper bounds: roughly, by~designing a~fast algorithm
		for checking satisfiability of~circuits, one gets a~lower bound for this circuit
		class. Since then, a~number of~results of~this kind have been proved. For example,
		Jahanjou et al. (ICALP 2015) and Carmosino et al. (ITCS 2016) proved that if~\NSETH{} fails, then $\E^{\NP}$ has series-parallel circuit size~$\omega(n)$.

		One can also derive nonuniform lower bounds
		from nondeterministic uniform lower bounds.
		Perhaps the most well-known example is~the Karp--Lipton theorem (STOC 1980): if $\Sigma_2 \neq \Pi_2$, then $\NP \not\subset \Ppoly$. Some recent examples include the following.
		Nederlof (STOC 2020) proved
		a~lower bound on~the
		matrix multiplication tensor rank
		under an~assumption that \TSP{} cannot be~solved faster than in~$2^n$ time.
		Belova et al. (SODA 2024)
		proved that there exists an~explicit polynomial family of~arithmetic circuit size $\Omega(n^{\delta})$, for any $\delta > 0$, assuming that \MAX{}-3-\SAT{} cannot be~solved faster than in~$2^n$ nondeterministic time. Williams (FOCS 2024) proved
		an~exponential lower bound for $\ETHR{} \circ \ETHR{}$ circuits under the \problem{Orthogonal Vectors} conjecture.
		Under an~assumption that \problem{Set Cover} problem cannot be~solved
		faster than in~$2^n$, Bj{\"{o}}rklund and~Kaski (STOC 2024)
		and Pratt (STOC 2024) constructed an~explicit tensor with superlinear rank.
		Whereas all the lower bounds above are proved under strong assumptions
		that might eventually be~refuted,
		the revealed connections are of~great interest and may still give further
		insights: one may be~able to~weaken the used assumptions or~to~construct generators
		from other fine-grained reductions.

		In~this paper, we~continue developing this line of~research and
		show how uniform nondeterministic lower bounds
		can be~used to~construct generators of~various types of~combinatorial
		objects that are notoriously hard to~analyze:
		Boolean functions of~high circuit size, matrices of~high rigidity,
		and tensors of~high rank.
		Specifically, we~prove the following.
		\begin{itemize}
			\item If, for some $\varepsilon$~and~$k$, $k$-\SAT{} cannot be~solved in~input-oblivious co-nondeterministic time $O(2^{(1/2+\varepsilon)n})$,
			then there exists a~monotone Boolean function family in \co\NP{} of monotone circuit size $2^{\Omega(n / \log n)}$. Combining this with the result above, we get win-win circuit lower bounds: either $\E^{\NP{}}$ requires series-parallel circuits
			of~size $\omega(n)$ or~$\coNP$ requires monotone circuits of~size $2^{\Omega(n / \log n)}$.
			\item If, for all $\varepsilon>0$, \MAX{}-3-\SAT{} cannot be solved in co-nondeterministic time $O(2^{(1 - \varepsilon)n})$,
			then there exist small families of~matrices
			with rigidity exceeding the best known constructions
			as~well~as small families of~three-dimensional tensors
			of~rank $n^{1+\Delta}$, for some $\Delta>0$.
		\end{itemize}
	\end{abstract}

	\clearpage

    \tableofcontents

	\clearpage

    \section{Complexity Lower Bounds}
	Finding the~minimum time required to~solve a~given computational problem is a~central question in~computational complexity.
	Answering such a~question for a~particular problem involves
	proving a~complexity lower bound, that~is, showing that
	no~fast algorithm can solve this problem. While the Time Hierarchy Theorem~\cite{hartmanis1965computational,DBLP:journals/jacm/HennieS66}
	guarantees that there are problems in~\P{} that cannot be~solved in~time $O(n^k)$, for any $k>1$, we~have no~superlinear lower bounds for specific problems. For example, for~\SAT{}, one of~the most important \NP{}-complete problems, we~have no~algorithms working significantly faster than a~brute force approach and at~the same time have no~methods of~excluding
	a~possibility that it~can be~solved in~linear time.

	\subsection*{Conditional Lower Bounds}
	As~unconditional complexity lower bounds remain elusive, the classical complexity theory allows one to~prove \emph{conditional} lower bounds of~the following form:
	if a~problem~$A$ cannot be~solved in~polynomial time,
	then $B$~also cannot be~solved in~polynomial time. Such results
	are proved via reductions that are essentially algorithms: one shows how to~transform an~instance of~$A$ into an~instance of~$B$. Nowadays, hundreds of~such reductions between various \NP{}-hard problems are known. For instance, if \SAT{} cannot
	be~solved in~polynomial time, then the~\problem{Hamiltonian Cycle} problem also cannot be solved in~polynomial time.

	In~a~recently emerged area of~\emph{fine-grained complexity},
	one aims to~construct tighter reductions between problems showing that even a~tiny improvement of an~algorithm for one
	of~them automatically leads to~improved algorithms for the other one. For example, as~proved by~Williams~\cite{DBLP:journals/tcs/Williams05}, if \SAT{} cannot be~solved in~time $O(2^{(1-\varepsilon)n})$, for~any~$\varepsilon > 0$,
	then the~\problem{Orthogonal Vectors} problem cannot be solved in~time~$O(n^{2 - \varepsilon})$, for~any~$\varepsilon > 0$. Again, many reductions of~this form
	have been developed in~recent years. We~refer the reader
	to~a~recent survey by~Vassilevska Williams~\cite{williams2018some}.

	\subsection*{Circuit Lower Bounds}
	One of~the reasons why proving complexity lower bounds
	is~challenging is~that an~algorithm (viewed as a~Turing machine or a~RAM machine) is a~relatively complex object: it~has a~memory, may contain loops, function calls (that may in~turn be~recursive). A~related computational model of~\emph{Boolean circuits} has a~much simpler structure (a~straight-line program) and at the same time is~powerful enough to~model algorithms:
	if a~problem can be~solved by~algorithms in~time~$T(n)$,
	then it~can also be~solved by~circuits of~size $O(T(n)\log T(n))$~\cite{DBLP:journals/jacm/PippengerF79}.
	It~turns out that proving circuit lower bounds is also challenging: while it~is not difficult to~show that almost all Boolean functions can be~computed by~circuits of~exponential size only (this was proved by~Shannon~\cite{DBLP:journals/bstj/Shannon49} back in~1949),
	for no~function from \NP{}, we can currently exclude the~possibility that it~can be~computed by~circuits of~linear
	size~\cite{DBLP:conf/stoc/Li022,DBLP:conf/focs/FindGHK16}.
	Strong lower bounds are only known for restricted models
	such as~monotone circuits, constant-depth circuits, and formulas.
	Various such unconditional lower bounds can be~found in~the book
	by~Jukna~\cite{DBLP:books/daglib/0028687}.

	An~important difference between algorithms and circuits is~that
	algorithms represent a~\emph{uniform} model of~computation (an~algorithm is a~program that needs to~process instances
	of~all possible lengths), whereas circuits are \emph{nonuniform}: when saying that a~problem can be~solved by~circuits, one usually means that there is an~infinite collection of~circuits, one circuit for every possible input length, and different circuits in~this collection can, in~principle, implement different programs. This makes the circuit model strictly more powerful than algorithms:
	on~the one hand, every problem solved by~algorithms can be~solved by~circuits of~roughly the same size; on~the other hand, it~is
	not difficult to~come~up with a~problem of~small circuit size
	that cannot be~solved by~algorithms.

	\subsection*{Connections Between Lower and Upper Bounds}
	Intuitively, it~seems that proving complexity upper bounds
	should be~easier than proving lower bounds.
	This intuition is~well supported by~a~much higher number
	of~results on~algorithms compared to~the number of~results on~lower bounds.
	Indeed, to~prove an~upper bound on~the complexity of a~problem, one
	designs an~algorithm for the problem and analyzes~it.
	Whereas to~prove a~complexity lower bound, one needs to~reason
	about a~wide range of~fast algorithms (or~small circuits) and to~argue that none of~them is~able to~solve the problem at~hand.
	Perhaps surprisingly, the tasks of~proving lower and upper complexity bounds are connected to~each other.
	A~classical example is
	Karp--Lipton theorem~\cite{DBLP:conf/stoc/KarpL80} stating that if \P{} = \NP{}, then \EXP{} requires circuits of size~$\Omega(2^n / n)$.
	More recently, Williams~\cite{DBLP:journals/siamcomp/Williams13} established a~deep connection between upper bound for \problem{Circuit Sat} and circuit lower bounds.
	Extending his results, Jahanjou, Miles and Viola~\cite{DBLP:journals/iandc/JahanjouMV18} proved that
	if~\NSETH{} is false (meaning that \UNSAT{} can be~solved fast with nondeterminism), then $\E{}^{\NP{}}$ requires series-parallel Boolean circuits of~size~$\omega(n)$. Such results show how
	to~derive nonuniform lower bounds (that~is, circuit lower bounds)
	from uniform \emph{upper} bounds (algorithm upper bounds).

	Even though one~can simulate an~algorithm using circuits with~slight overhead, the~converse is~not true as~there are undecidable languages of~low circuit complexity.
	Recently,~\cite{DBLP:conf/soda/BelovaKMRRS24} showed results analogous to~those presented herein, particularly on~deriving a~nonuniform lower bound from a~non-randomized uniform lower bound.
	Specifically, they proved that if \MAX{}-$k$-\SAT{} cannot be solved in co-nondeterministic time~$O(2^{(1 - \varepsilon)n})$, for any~$\varepsilon > 0$, then, for any~$\delta > 0$, there exists an~explicit polynomial family that cannot be computed by~arithmetic circuits of size~$O(n^\delta)$.
	Also, Williams~\cite{DBLP:conf/focs/Williams24} proved that
	if~the \problem{Orthogonal Vectors} conjecture (\OVC{}) holds, then \problem{Boolean Inner Product} on~$n$-bit vectors cannot be~computed by $\ETHR{} \circ \ETHR{}$~circuits
	of~size~$2^{\varepsilon n}$, for some $\varepsilon > 0$.
	Combined with the result above (since \OVC{} is weaker than \NSETH{}), it~immediately leads to~win-win circuit lower bounds: if~\NSETH{} fails, we~have a~lower bound for series-parallel circuits, otherwise we~have a~strong lower bound for
	$\ETHR{} \circ \ETHR{}$ circuits.

	\subsection{Our Contribution}
	In~this paper, we~derive a~number of~nonuniform lower bounds
	from uniform nondeterministic lower bounds. Our lower bounds
	apply for various objects that are notoriously hard to~analyze:
	Boolean functions of~high monotone circuit size, high rigidity matrices,
	and high rank tensors.
	For circuits, we~get
	a~win-win situation similar to~the one by~Williams.

	Our first result
	shows how to~get $2^{\Omega(n/\log n)}$ monotone circuit lower bounds
	(improving best known bounds of~the form $2^{\Omega(\sqrt n)}$)
	under an~assumption that \SAT{} requires co-nondeterministic time $O(2^{(1/2+\varepsilon)n})$ if~the verifier is~given a~proof that depends
	on~the length of~the input only.

	\begin{restatable}{theorem}{thmmonotone}
		\label{theorem:monotone_conp}
		If, for some $\varepsilon > 0$ and~$k \in \mathbb{Z}_{\geq 3}$,
		$k$-\SAT{} cannot be~solved in input-oblivious co-nondeterministic time~$O(2^{(1/2+\varepsilon)n})$,
		then there exists a~monotone Boolean function family in~$\coNP$ of~monotone circuit size~$2^{\Omega(n/\log n)}$.
	\end{restatable}

	As the~previous theorem uses weaker assumption than \NSETH{},
	we~get the following corollary.

	\begin{restatable}{corollary}{crlymonotone}
		\label{corollary:monotone_nseth}
		If~\NSETH{} holds, there exists a~monotone Boolean function family in~\coNP{}  
		with monotone circuit size~$2^{\Omega(n/\log n)}$.
	\end{restatable}

	Combining this with circuit lower bounds that follow from the negation of~\NSETH{} due to~\cite{DBLP:journals/iandc/JahanjouMV18,DBLP:conf/innovations/CarmosinoGIMPS16}, leads to~win-win circuit lower bounds.
	\begin{restatable}{corollary}{crlMonotone}
		\label{crl:monotone}
        At~least one of~the following two circuit lower bounds holds:
        \begin{enumerate}
            \item $\E^{\NP}$ requires series-parallel circuits of~size~$\omega(n)$;
            \item There exists a~monotone Boolean function family in~\coNP{}
            of~monotone circuit size~$2^{\Omega(n/\log n)}$.
        \end{enumerate}
	\end{restatable}
	Each of~these lower bounds is~far from what is~currently known:
	the best known circuit lower bound for $\E^{\NP}$ is~$3.1n-o(n)$~\cite{DBLP:conf/stoc/Li022},
	whereas the best known monotone circuit lower bound for \coNP{}
	is $2^{\Omega(\sqrt n)}$ (see Section~\ref{section:explicit}).

%

	Another result in~this direction demonstrates how to derive circuit lower bounds from~$\NETH$ (which asserts that $3$-\SAT{} cannot be~solved in co-nondeterministic time~$2^{o(n)}$).
	Specifically, we establish lower bounds for the classes~$Q_t^{n}$.
	\begin{definition}
    	Let~$Q_t^{n}$ denote the set of all Boolean functions~$f$ over~$n$ variables such that for any $x^1, \ldots, x^t \in f^{-1}(0)$, there exists an $i \in [n]$ for which $x^1_i = \ldots = x^t_i = 0$.
	\end{definition}
	These sets have been studied in~secure multiparty computation; see~\cite{DBLP:conf/podc/HirtM97,DBLP:conf/wdag/FitziM98,DBLP:journals/joc/HirtM00} for references.
	They have also been viewed from the complexity-theoretic perspective~\cite{DBLP:conf/crypto/CohenDIKMRR13,DBLP:journals/toc/KozachinskiyP22}.
	We denote by~$\THR_{a}^{b} \colon \{0, 1\}^{b} \to \{0, 1\}$ the function that evaluates to~one if~and only~if the input contains at least~$a$ ones.
	It~turns out the class~$Q_t^{n}$ is exactly the class of functions computable by~$\THR_{l + 1}^{lt + 1}$ gates for arbitrary~$l \ge 1$ where the negations are not permitted (even at~the leaves). The following lemma
    is~a~combination of~Lemmas~5.2 and 5.3 from~\cite{DBLP:conf/crypto/CohenDIKMRR13} and Theorem~3
    from~\cite{DBLP:journals/toc/KozachinskiyP22}. In~the statement, by~$C \le f$ we~mean that $C(x) \le f(x)$,
    for all~$x$.

	\begin{lemma}
        \label{statement:check_q_t}
    	The set~$Q_t^{n}$ is~equal to~the set of functions~$f$ for which there exists a~circuit~$C$, composed only of~$\THR_{l + 1}^{lt + 1}$ gates for arbitrary~$l \ge 1$, such that $C \le f$.
	\end{lemma}
	Our result indicates that
    $Q_t^{n}$ contains functions that are exponentially hard
    to~represent using $\THR_{l + 1}^{lt + 1}$ gates only
    while being easy to~compute by~regular circuits.
	\begin{restatable}{theorem}{thmthreshold}
        \label{theorem:threshold}
    	If~$\NETH$ holds, then for any $t = \omega(1)$ and infinitely many~$n = \omega(t)$, there exists a function $f \in Q_t^n \cap \coNP_{/\poly}$ such that any circuit~$C \le f$, composed of~$\THR_{l + 1}^{lt + 1}$ gates, satisfies $\size(C) = 2^{\Omega(n)}$.
	\end{restatable}
	Thus, although the class~$Q_t^n$ may appear relatively simple, it contains functions that are exceedingly difficult to compute using only~$\THR_{l + 1}^{lt + 1}$ gates but which can be computed by almost linear circuits.
	\Cref{theorem:threshold} consequently yields another win-win lower bound: if~$\NETH$ fails, then $\E^{\NP}$ cannot be~computed by linear-size circuits~\cite{DBLP:journals/jacm/ChenRTY23}.

	\begin{corollary}
	At least one of the following two circuit lower bounds must hold:
	\begin{enumerate}
		\item $\E^{\NP}$ requires circuits of size~$\omega(n)$;
		\item For any $t = \omega(1)$, there exists a function $f \in Q_t^{n} \cap \Ppoly$ such that any circuit~$C \le f$, composed only of~$\THR_{l + 1}^{lt + 1}$ gates for arbitrary~$l \ge 1$ over variables, has size~$2^{\Omega(n)}$. Moreover, $f$ can be~computed in linear time.
	\end{enumerate}
	\end{corollary}

	Our second result shows how to~construct small families of~matrices
	with rigidity exceeding the best known constructions
	under an~assumption that \MAX{}-3-\SAT{} requires co-nondeterministic time nearly~$2^n$.

	\begin{restatable}{theorem}{thmHighRigidity} \label{thm:high_rigidity}
		If, for every $\varepsilon > 0$, \MAX{}-3-\SAT{} cannot be~solved in co-nondeterministic time $O(2^{(1 - \varepsilon)n})$,
		then, for all $\delta > 0$, there is a~generator $g  \colon \{0, 1\}^{\log^{O(1)} k} \to  \mathbb{F}^{k \times k}$ computable in~time polynomial in~$k$ such that, for infinitely many~$k$, there exist a~seed~$s$ for which $g(s)$ has $k^{\frac 1 2 - \delta}$-rigidity~$k^{2 - \delta}$.
    \end{restatable}

    Our third result extends the second result by~including high-rank tensors.

	\begin{restatable}{theorem}{thmHighRankTensorsOmega} \label{thm:high_rank_tensors_omega}
		If, for any $\varepsilon > 0$, \MAX{}-3-\SAT{} cannot be~solved in co-nondeterministic time $O(2^{(1 - \varepsilon)n})$, then, for all $\delta > 0$ and some $\Delta > 0$, there are two generators $g_1  \colon \{0, 1\}^{\log^{O(1)} k} \to  \mathbb{F}^{k \times k}$ and $g_2  \colon \{0, 1\}^{\log^{O(1)} k} \to \mathbb{F}^{k \times k \times k}$ computable in~time polynomial in $k$ such that, for infinitely many~$k$, at~least one of~the following is~satisfied:
	       \begin{itemize}
	           \item $g_1(s)$ has $k^{1 - \delta}$-rigidity~$k^{2 - \delta}$, for some~$s$;
			   \item $\rank(g_2(s))$ is at least $k^{1 + \Delta}$, for some~$s$.
	       \end{itemize}
	\end{restatable}

	It~is worth noting that~\cite{DBLP:conf/soda/BelovaKMRRS24} showed circuit lower bounds under the same assumption: if \MAX{}-$k$-\SAT{} cannot be solved in co-nondeterministic time~$O(2^{(1 - \varepsilon)n})$ for any~$\varepsilon > 0$, then for any~$\delta > 0$, there exists an~explicit polynomial family that cannot be computed by arithmetic circuits of size~$O(n^\delta)$.

	The best known lower bounds for the size of~depth-three circuits computing an~explicit Boolean function is~$2^{\Omega(\sqrt n)}$~\cite{DBLP:journals/acr/Hastad89,DBLP:journals/jacm/PaturiPSZ05}.
	Proving a~$2^{\omega(\sqrt n)}$ lower bound for this restricted circuit model remains
	a~challenging open problem and it~is known that a~lower bound as~strong~as $2^{\omega(n/\log\log n)}$ would give an~$\omega(n)$ lower bound for unrestricted circuits
	via Valiant's reduction~\cite{DBLP:conf/mfcs/Valiant77}.
	One way of~proving better depth-three circuit lower bounds is~via \emph{canonical circuits} introduced by~Goldreich and~Wigderson~\cite{DBLP:series/lncs/0001W20}.
	They are~closely related to~rigid matrices: if~$T$ is~an~$n \times n$ matrix of~$r$-rigidity $r^3$, then~the~corresponding bilinear function requires canonical circuits of~size $2^{\Omega(r)}$~\cite{DBLP:series/lncs/0001W20}.
	Goldreich and~Tal~\cite{DBLP:journals/cc/GoldreichT18} showed that a~random Toeplitz matrix has~$r$-rigidity $\frac{n^3}{r^2 \log n}$, which implies a~$2^{\Omega(n^{3/5})}$ lower bound on~canonical depth-three circuits for~an~explicit function.
	By~substituting $n^{2/3 - \delta}$-rigidity for~some $\delta > 0$
	in~\Cref{thm:high_rank_tensors_omega}, one gets the~following result.
	\begin{restatable}{corollary}{crlCanonical}
		If, for every $\varepsilon > 0$, \MAX{}-3-\SAT{} cannot be~solved in~co-nondeterministic time $O(2^{(1 - \varepsilon)n})$, then, for any $\delta > 0$, one can construct an~explicit family of $2^{\log^{O(1)}n}$ functions such that,
		for infinitely many~$n$, at~least one of~them is
		either bilinear and~requires canonical circuits of~size $2^{\Omega(n^{2/3 - \delta})}$ or~trilinear and~requires arithmetic circuits of~size $\Omega(n^{1.25})$.
	\end{restatable}
	This conditionally improves the~recent result of~Goldreich~\cite{DBLP:journals/cc/Goldreich22}, who presented an~$O(1)$-linear function that requires canonical depth-\emph{two} circuits of~size $2^{\Omega(n^{1 - \varepsilon})}$, for~every $\varepsilon > 0$.
	Moreover, every bilinear function can be~computed by~canonical circuits of~size $2^{O(n^{2/3})}$, so~the~lower bound is~almost optimal and~conditionally addresses Open Problem~6.5 from~\cite{DBLP:journals/cc/GoldreichT18}.

	\subsection{Known Explicit Constructions}
	\label{section:explicit}
	In~this section, we~review known constructions of~combinatorial objects
	(functions of~high monotone circuit size, matrices of~high rigidity, and
	tensors of~high rank).

	\subsubsection*{Monotone Functions}
	For monotone \NP{}-problems (like \problem{Clique}, \problem{Matching}, \problem{Hamiltonian Cycle}), it~is natural
	to~ask what is~their monotone circuit size.
	A~celebrated result by~Razborov~\cite{razborov1985lower} is~a~lower bound of~$n^{\Omega(\log n)}$ on~monotone circuit size obtained by~the approximation method (which was recently improved to $2^{n^{\Omega(1)}}$~\cite{DBLP:journals/eccc/CavalarGRS025}).
	Subsequently, Andreev~\cite{andreev1985method} proved a~$2^{n^{1/8 - o(1)}}$ lower bound for another explicit monotone function.
	Following the work of~\cite{DBLP:journals/combinatorica/AlonB87,andreev1987method,DBLP:journals/combinatorica/Jukna99,DBLP:conf/stoc/HarnikR00}, in~2020, Cavalar, Kumar, and~Rossman~\cite{DBLP:journals/algorithmica/CavalarKR22} achieved the best-known lower bound of~$2^{n^{1/2 - o(1)}}$.
	Recently, another approach for~proving monotone circuit lower bounds was developed using lower bounds from Resolution proofs and lifting theorem~\cite{DBLP:journals/toc/GargGK020}.
	Specifically, if an unsatisfiable formula~$F$ is hard to~refute in~the resolution proof system, then a~monotone function associated with~$F$ has large monotone circuit complexity.
	In~this manner, following the~work of~\cite{DBLP:journals/toc/GargGK020,DBLP:conf/innovations/GoosKRS19,DBLP:conf/innovations/LovettMMPZ22}, the~lower bound of~$2^{\Omega(\sqrt{n})}$ was~also achieved.
	Just recently, Blasiok and Meierh{\"{o}}fer~\cite{DBLP:conf/coco/BlasiokM25} established the lower bound~$2^{\Omega(\sqrt{n})}$ for the~\problem{Clique} problem.

	It is worth to mention, that several monotone functions from~$\P$ are known to require exponential monotone circuit size~\cite{razborov1985lower,DBLP:journals/combinatorica/Tardos88,DBLP:journals/jacm/RazW92}.
	Therefore, the gap between monotone circuit size and general circuit size could be exponential.

    Proving a~$2^{\omega(n^{1/2})}$ lower bound remains a~challenging open problem (whereas a~lower bound $2^{\Omega(n)}$ was recently proved by~Pitassi and Robere~\cite{DBLP:conf/stoc/PitassiR17} for monotone \emph{formulas}). Our \Cref{corollary:monotone_nseth} establishes
    a~stronger lower bound under an~assumption that \NSETH{} holds.

	\subsubsection*{Matrix Rigidity}
    A~matrix~$M$ over a~field~$\mathbb F$ has \emph{r-rigidity~s} if~for any matrices $R, S$ over
    a~field~$\mathbb{F}$ such that $M = R + S$ and $\rank(R) \le r$, $S$~has at~least $s$~nonzero entries.
	That~is, one needs to~change at~least $s$ elements in~$M$
	to~change its rank down to~at~most~$r$.
    The concept of~rigidity was introduced by~Valiant~\cite{DBLP:conf/mfcs/Valiant77} and Grigoriev~\cite{grigor1980application}.
    It has striking connections to~areas such~as computational complexity~\cite{DBLP:journals/fttcs/Lokam09,DBLP:conf/stoc/AlmanW17, DBLP:conf/focs/AlmanC19, DBLP:conf/innovations/GolovnevKW21}, communication complexity~\cite{DBLP:journals/cc/Wunderlich12}, data structure lower bounds~\cite{DBLP:conf/stoc/DvirGW19, DBLP:conf/innovations/RamamoorthyR20}, and~error-correcting codes~\cite{DBLP:journals/cc/Dvir11}.

    Valiant~\cite{DBLP:conf/mfcs/Valiant77} proved that~if a~matrix~$M$ has~$\varepsilon n$-rigidity~$n^{1 + \delta}$ for some~$\varepsilon, \delta > 0$, then the bilinear form of~$M$ cannot be~computed by~arithmetic circuits of~size~$O(n)$ and~depth~$O(\log n)$.
    Following Razborov~\cite{razborov89}, Wunderlich~\cite{DBLP:journals/cc/Wunderlich12} proved that the~existence of strongly-explicit matrices with~$2^{(\log \log n)^{\omega(1)}}$-rigidity~$\delta n^2$, for some~$\delta > 0$, implies the existence of a~language that does~not belong to~the communication complexity analog of~$\PH$.
	Although it~is known~\cite{DBLP:conf/mfcs/Valiant77} that~for any~$r$ almost every $n \times n$ matrix has~$r$-rigidity~$\Omega(\frac{(n-r)^2}{\log n})$ over algebraically closed fields, obtaining an~explicit constructions of~rigid matrices remains a~long-standing open question.
    Many works have aimed~at finding explicit or~semi-explicit rigid matrices~\cite{DBLP:journals/combinatorica/Friedman93,pudlak1991computation, DBLP:journals/ipl/ShokrollahiSS97, DBLP:conf/stoc/AlmanW17, DBLP:journals/toc/DvirE19, DBLP:conf/focs/AlmanC19, DBLP:journals/toc/DvirL20, DBLP:journals/toct/VolkK22, DBLP:journals/jacm/BhargavaGKM23, DBLP:journals/siamcomp/BhangaleHPT24}.
    Also, a~recent line of~works establishes a~connection between
    the \problem{Range Avoidance} problem and the construction of matrices with high
    rigidity~\cite{DBLP:conf/focs/Korten21, DBLP:conf/approx/GuruswamiLW22, DBLP:conf/approx/GajulapalliGNS23, DBLP:conf/stoc/0001HR24,DBLP:conf/stoc/Li24}.

	Small explicit\footnote{A~matrix or family of matrices is~called explicit if it~is polynomial-time constructible.} families of rigid matrices can be~used to prove arithmetic circuits lower bounds~\cite{DBLP:conf/mfcs/Valiant77}.
    The best known polynomial-time constructible matrices have~$r$-rigidity~$\frac {n^2} r \log(n/r)$ for~any~$r$, which was proved by Shokrollahi, Spielman and Stemann~\cite{DBLP:journals/ipl/ShokrollahiSS97}.
	Goldreich and~Tal~\cite{DBLP:journals/cc/GoldreichT18} proved that a~random $n \times n$ Toeplitz matrix over~$\mathbb{F}_2$ (i.e., a~matrix of~the~form $A_{i,j} = a_{i - j}$ for random bits $a_{-(n - 1)}, \ldots, a_{n - 1}$) has $r$-rigidity $\frac{n^3}{r^2 \log n}$ for~$r \ge \sqrt{n}$.
	However, the~size of~that family is~exponential in~$n$.
	Our Theorem~\ref{thm:high_rigidity} demonstrates that, under the assumption that \MAX{}-$3$-\SAT{} is~hard, for~any $\delta > 0$, for~infinitely many~$n$ one can construct a~$2^{\log^{O(1)}n}$-sized family of~$n \times n$ matrices with at~least one having $n^{1 / 2 - \delta}$-rigidity $n^{2 - \delta}$.

	This result is~still far from~the~regime where circuit lower bounds can~be derived via~Valiant's result, but it~strictly improves the~polynomial-time construction~\cite{DBLP:journals/ipl/ShokrollahiSS97} for~any $r < \sqrt{n}$ and~improves the result of~Goldreich and~Tal~\cite{DBLP:journals/cc/GoldreichT18} by~substantially reducing the~family size while maintaining the~same rigidity for~$r \approx \sqrt{n}$.
	An open question remains as to~whether explicit constructions of~rigid matrices exist in~the class $\P^{\NP{}}$~\cite{DBLP:journals/corr/abs-2009-09460}.
	The construction provided by~Goldreich and~Tal~\cite{DBLP:journals/cc/GoldreichT18} lies in~$\E^{\NP{}}$.
	Following the work of Alman and~Chen~\cite{DBLP:conf/focs/AlmanC19},~\cite{DBLP:journals/siamcomp/BhangaleHPT24} established that there exists a constant $\delta > 0$ such that one can construct $n \times n$ matrices with $2^{\log n / \Omega(\log \log n)}$-rigidity $\delta n^2$ in~$\mathrm{FNP}$.
	Subsequently, Chen and~Lyu~\cite{DBLP:conf/stoc/0001L21} demonstrated a method for constructing highly rigid matrices, proving that there exists a constant $\delta > 0$ such that one can construct $n \times n$ matrices with $2^{\log^{1 - \delta} n}$-rigidity $(1 / 2 - \exp(-\log^{2 / 3 \cdot \delta} n)) \cdot n^2$ in~$\P^{\NP}$.
	More recently, Alman and~Liang~\cite{DBLP:conf/stoc/AlmanL25} showed that the Walsh-Hadamard matrix~$H_n$ has $c_1 \log n$-rigidity $n^2 \left( \frac{1}{2} - n^{c_2} \right)$ for some constants $c_1, c_2 > 0$.
	Our construction, in~the class $\DTIME[2^{\log^{O(1)}n}]^{\NP}$, produces matrices with $n^{\frac{1}{2} - \delta}$-rigidity~$n^{2 - \delta}$ for~any $\delta > 0$, under~the condition that \MAX{}-3-\SAT{} is~hard.

	\subsubsection*{Tensor Rank and Arithmetic Circuits}
	Proving arithmetic circuit lower bounds is~another important
	challenge in~complexity theory.
	An~arithmetic circuit over a~field~$\mathbb F$ uses as~inputs formal variables and field elements and computes in~every gate
	either a~sum or a~product.
	As proved~by Strassen~\cite{strassen1973berechnungskomplexitat, DBLP:journals/tcs/Strassen75} and Baur~and~Strassen~\cite{DBLP:journals/tcs/BaurS83}, computing $\sum_{i = 1}^{n} x_{i}^{n}$ requires arithmetic circuits of size~$\Omega(n \log n)$, provided~$n$ does not divide the~characteristic of~$\mathbb{F}$.
	Raz~\cite{DBLP:journals/siamcomp/Raz03} further established that arithmetic circuits with bounded coefficients require $\Omega(n^2 \log n)$ gates to perform matrix multiplication over $\mathbb{R}$ or $\mathbb{C}$, following the work in~\cite{DBLP:journals/siamcomp/RazS03}.
	However, no~superlinear lower bounds are known for polynomials
	of~constant degree.
	For constant-depth arithmetic circuits over fields of characteristic~$2$, exponential lower bounds are known~\cite{razborov1987lower,DBLP:conf/stoc/Smolensky87}.
	For other finite characteristics, exponential lower bounds are known only for depth~$3$~\cite{DBLP:conf/stoc/GrigorievK98,DBLP:conf/focs/GrigorievR98}.
	For characteristic~$0$, the best lower bound for depth~$3$ is $\Omega(n^{2 - \varepsilon})$~\cite{DBLP:conf/coco/ShpilkaW99}.
	For~characteristic~$0$, Limaye, Srinivasan, and~Tavenas~\cite{DBLP:journals/jacm/LimayeST25} proved the~first superpolynomial lower bounds for constant-depth circuits.
	We~refer to~\cite{DBLP:journals/jacm/LimayeST25,DBLP:conf/icalp/AmireddyGK0T23,DBLP:journals/toct/BhargavDS24,DBLP:conf/coco/000124} and the~references therein for~recent advances in~lower bounds for~small-depth algebraic circuits.

	Matrix multiplication is~one of~the fundamental problems
	whose arithmetic circuit size is~of~great interest.
	While many highly nontrivial algorithms for~it
	are known (starting from Strassen~\cite{strassen1969gaussian}),
	we~still do~not have superlinear lower bounds on~its arithmetic circuit complexity. Proving such lower bounds
	is closely related to the problem of~determining the rank of tensors.
	A~$d$-dimensional tensor is said to have rank~$q$ if it can be expressed as~a~sum of~$q$ rank-one tensors.
	Here, a~rank-one $d$-dimensional tensor is a~tensor of the form $u_1 \otimes \dots \otimes u_d$, where $\otimes$ stands for a tensor product.
	By a~multiplication tensor, we mean a~tensor of size $n^2 \times n^2 \times n^2$ (formally defined in~\Cref{sec:matrices_and_tensors}).
	Establishing an upper bound for the rank of the multiplication tensor provides a~means
	of~proving upper bounds for matrix multiplication via the laser method~\cite{DBLP:conf/focs/Strassen86}.
	Moreover, proving a~lower bound for the tensor rank would yield superlinear lower bounds for arithmetic circuits computing the polynomial defined by that tensor.

	Therefore, proving lower bounds on~the tensor rank provides a~path to~proving lower bounds for arithmetic circuits.
	For the rank of~the matrix multiplication tensor, Bshouty~\cite{DBLP:journals/siamcomp/Bshouty89} and Bläser~\cite{DBLP:conf/focs/Blaser99} proved a~lower bound $2.5n^2 - \Theta(n)$.
	Subsequently, Shpilka~\cite{DBLP:journals/siamcomp/Shpilka03} improved the bound to $3n^2 - o(n^2)$ over~$\mathbb{F}_2$.
	The bound $3n^2 - o(n)$ was later achieved by~Landsberg~\cite{DBLP:journals/siamcomp/Landsberg14} over arbitrary fields and further slightly improved by Massarenti and Raviolo~\cite{MASSARENTI20134500,MASSARENTI2014369}.
	Alexeev, Forbes and Tsimerman~\cite{DBLP:conf/coco/AlexeevFT11} constructed explicit $d$-dimensional tensors with rank $2 n^{\left\lfloor \frac{d}{2} \right\rfloor} + n - \Theta(d \log n)$, thus improving the lower bounds on~high-dimensional tensors.
	Nevertheless, superlinear size lower bounds for constant-degree polynomials remain unknown.
	Additionally, H{\aa}stad~\cite{DBLP:journals/jal/Hastad90} established that determining the rank of a~$d$-dimensional tensor is~\NP{}-hard for any $d \geq 3$.
	Consequently, a~major open problem is to~construct
	an~explicit family of $d$-dimensional tensors with rank at~least $n^{\left\lfloor \frac{d}{2} \right\rfloor + \varepsilon}$ for some $\varepsilon > 0$ and $d \geq 3$.

	Our \Cref{thm:high_rank_tensors_omega} shows that, under an~assumption that \MAX{}-3-\SAT{} cannot be~solved fast co-nondeterministically, one gets
	an~explicit $2^{\log^{O(1)}n}$-size family of $n \times n$-matrices and $n \times n \times n$-tensors, such that, for any $\delta > 0$ and some $\Delta > 0$, at~least one of~the matrices has $n^{1 - \delta}$-rigidity~$n^{2 - \delta}$ or one of~the tensors has rank $n^{1 + \Delta}$.
	Furthermore, we establish a~trade-off between matrix rigidity and tensor rank, see~\Cref{thm:high_rank_tensors}.
	Other results for proving lower bounds on~tensor rank under certain assumptions are~known.
	Nederlof~\cite{DBLP:conf/stoc/Nederlof20} proved that, if for~any $\varepsilon > 0$, the~bipartite \problem{Traveling Salesman} problem cannot be~solved in~time $2^{(1 - \varepsilon)n}$, then the~matrix multiplication tensor has superlinear rank.
	Additionally, Bj{\"{o}}rklund and~Kaski~\cite{DBLP:conf/stoc/BjorklundK24} recently proved that if, for~any $\varepsilon > 0$, there exists a~$k$ such that the~$k$-\problem{Set Cover} problem cannot be~solved in~time $O(2^{(1 - \varepsilon)n} | \mathcal{F}|)$, then there~is an~explicit tensor with superlinear rank, where~$\mathcal{F}$ is a~family of~subsets of~$[n]$, each of~size at~most~$k$.
	Pratt~\cite{DBLP:conf/stoc/Pratt24} improved this result, showing that under~the~same conjecture there exists an~explicit tensor of~shape $n \times n \times n$ and~rank at~least $n^{1.08}$.
	\cite{DBLP:conf/soda/BjorklundCHKP25} showed that if for every~$\varepsilon > 0$ \problem{Chromatic Number} problem cannot be~solved in~time $2^{(1 - \varepsilon) n}$, then there~exists an~explicit tensor of~superlinear rank.

	\subsection{Discussion and Open Problems}
	Many of~the lower bounds mentioned above are proved under various strong assumptions (on~the complexity of \SAT{}, \MAX{}-\SAT{}, \problem{Set Cover}, \problem{Chromatic Number}, \problem{Traveling Salesman}). They seem much stronger than merely $\P{} \neq \NP$ and might eventually be~refuted.
	Still, the revealed reductions between problems are of~great interest and may still yield further
	insights: one may be~able to~weaken the used assumptions or~to~construct generators
	from other fine-grained reductions. Moreover, as~with the case of~\NSETH{} assumption,
	one may be~able to~derive interesting consequences from both an~assumption and its negation
	leading to a~win-win situation. Below, we~state a~few open problems in~this direction.

	If one were to~partition $3$-\SAT{} in~\Cref{theorem:satov} into~$t$ parts, where~$t = \omega(1)$, then creating an~explicit function equivalent to~$t$-\OV{} would imply lower bounds for~that function under~\NETH{}.
	This approach would yield a~more advantageous win-win situation, as~the~assumption that~\NETH{} is~false gives~stronger lower bounds~\cite{DBLP:journals/jacm/ChenRTY23}.
	\begin{openproblem}
		Prove that if \NETH{} is true, then there exists an~explicit function of monotone complexity $2^{\omega(\sqrt{n} )}$.
	\end{openproblem}
	Moreover, the existence of small monotone circuits not only refutes~\NSETH{} but also establishes that~\UNSAT{} has input-oblivious proof size~$2^{o(n)}$ that can be verified in time~$2^{\frac{n}{2} + o(n)}$.
	Therefore, we believe that it may be possible to derive even stronger implications beyond merely refuting~\NSETH{}.

	Another question arises regarding tensor rank.
	Is~it possible to~construct a~small family of~tensors with~superlinear rank, assuming that~\MAX-3-\SAT{} cannot~be~solved efficiently in~co-nondeterministic time (this way, eliminating the~dependence on~rigidity from our results)?

	\begin{openproblem}
		Prove that, if,~for~any $\varepsilon > 0$, \MAX{}-3-\SAT{} cannot~be~solved in~co-nondeterministic time $O(2^{(1 - \varepsilon)n})$, then, for~all $\delta > 0$ there exists a~small family of~tensors of~size $k \times k \times k$ such that,
		for infinitely many~$k$, at~least one~of them has rank at~least $k^{1 + \delta}$.
	\end{openproblem}

	However, we do not know of any consequences of solving~\MAX{}-3-\SAT{} faster in the co-nondeterministic setting.
	This makes our assumption weak and raises the question of whether it~can be~refuted.

	\begin{openproblem}
		Prove that, for some $\varepsilon > 0$, \MAX{}-3-\SAT{} can be~solved
		in~co-nondeterministic time~$O(2^{(1 - \varepsilon)n})$.
	\end{openproblem}

	On~the other hand, to~get a~win-win situation, it~would be~interesting to~find
	nontrivial consequences of~the existence of~such an~algorithm.

	\begin{openproblem}
		Derive new circuit lower bounds from the existence of~an~algorithm
		solving \MAX{}-3-\SAT{}
		in~co-nondeterministic time~$O(2^{(1 - \varepsilon)n})$, for some $\varepsilon>0$.
	\end{openproblem}

	\subsection*{Structure of the Paper}
	The paper is organized as~follows.
	In~\Cref{sec:preliminaries}, we~introduce the~notation used throughout the~paper and provide the necessary background.
	In~\Cref{sec:core_concepts}, we~give an~overview of~the main proof ideas.
	In~\Cref{sec:circuit_lower_bounds}, we~establish the~win-win circuit lower bound.
	In~\Cref{sec:rigid_matrices}, we~construct rigid matrices under an~assumption that \MAX{}-3-\SAT{} cannot be~solved fast co-nondeterministically.
	In~\Cref{sec:high_rank_tensors}, we~construct either three-dimensional tensors with high~rank or~matrices with high~rigidity under the~same assumption.

	\section{Proof Ideas}
	\label{sec:core_concepts}
	In~this section, we~give high level ideas of~the main results.

    \subsection{Monotone Circuit Lower Bound}
    To~prove a~lower bound $2^{\Omega(n / \log n)}$ for monotone circuit size of \coNP{} under \NSETH{}, we~assume that all monotone
    functions from \coNP{} have monotone circuit size $2^{o(n / \log n)}$ and show how this can be~used to~solve \UNSAT{} nondeterministically in~less than $2^n$~steps.

    Given a~$k$-CNF~$F$, we~construct an~instance $(A_F, B_F)$ of~\OV{} of~size $2^{n/2}$ using \Cref{theorem:satov}.
    We~show (see \Cref{lemma:monotoneseparation}) that
    $(A_F, B_F)$ is a~yes-instance if~and only~if there exists
    a~monotone Boolean function separating $(A_F, \overline{B_F})$.
    Hence, one can guess a~small monotone circuit separating $(A_F, \overline{B_F})$ and verify that it~is correct. Overall,
    the resulting nondeterministic algorithm proceeds as~follows.

      \begin{quote}
            \begin{enumerate}
                \item In~time $O(n^22^{n/2})$, generate the sets $A_F$~and~$B_F$.
                \item Guess a~monotone circuit~$C_F$ of~size $O(2^{(1-\varepsilon)n/2})$.
                \item Verify that $C_F$~separates $(A_F, \overline{B_F})$. To~do this, check that $C_F(a)=1$, for every $a \in A_F$. Then, check that
                $C_F(b)=0$, for every $b \in \overline{B_F}$. The running time
                of~this step~is
                \[O(2^{(1-\varepsilon)n/2}\cdot |A_F| + 2^{(1-\varepsilon)n/2} \cdot |\overline{B_F}|)=O(2^{(1-\varepsilon)n/2} \cdot 2^{n/2})=O(2^{(1-\varepsilon/2)n}).\]
                \item If $C_F$ separates $(A_F, \overline{B_F})$, then
                $F \in \UNSAT{}$.
            \end{enumerate}
        \end{quote}

		This already shows how small monotone circuits could help to~break \NSETH{}, though it~does not provide a~single explicit function with this property.
		In the full proof in~\Cref{sec:circuit_lower_bounds},
		we~introduce such a~function. We~also use a~weaker assumption (than \NSETH{}).

	\subsection{Rigid Matrices and High Rank Tensors}
	To~construct matrices of~high rigidity under the~assumption that
	\MAX{}-$3$-\SAT{} cannot be~solved fast co-nondeterministically,
	we~proceed as~follows.
	Take a~$3$-CNF formula over $n$~variables and an~integer~$t$ and transform~it, using \Cref{theorem:max_3_sat_4_clique}, into a~$4$-partite $3$-uniform hypergraph $G$ with each part of~size $k = n^{O(1)}2^{\frac{n}{4}}$.
	The graph $G$~contains a~$4$-clique if~and only~if one can satisfy $t$~clauses of~$F$.
    We~show that checking whether $G$~has a~$4$-clique is~equivalent to~evaluating a~certain expression over three-dimensional tensors. We~then show that if~all slices
    of~these tensors have low rigidity, then one can solve \problem{$4$-Clique} on~$G$ in~co-nondeterministic time $O(k^{4-\varepsilon})$: to~achieve this, one guesses
    a~decomposition of a~matrix into a~sum of~a~low rank matrix
    and a~matrix with few nonzero entries.
    The idea is~that a~low rank matrix can be~guessed quickly
    as~a~decomposition into a~rank-one matrices (which are just products of~two vectors), whereas the second matrix can be~guessed quickly as~one needs to~guess the nonzero entries only.
    In~turn, this allows
    one to~solve \MAX{}-$3$-\SAT{} faster than~$2^n$ co-nondeterministically.
	Thus, this reduction is~a~generator of~rigid matrices: it takes a~$3$-CNF formula and outputs~a~matrix.
	The same idea can be~used to~generate either tensors with high rank or~matrices with high rigidity.

	\section{Preliminaries} \label{sec:preliminaries}
    For a~positive integer~$k$, $[k]=\{1, 2, \dotsc, k\}$, whereas for
    a~predicate~$P$, $[P]=1$ if~$P$~is true and $[P]=0$ otherwise (the Iverson bracket).
    For a~set~$S$ and an~integer~$k$, by $\binom{S}{k}$ we~denote the set
    of~all subsets of~$S$ of~size~$k$.

    \subsection{Boolean Circuits}\label{sec:spckts}
    For a~binary vector $v \in \{0,1\}^n$, by~$\overline{v} \in \{0,1\}^n$,
    we~denote the vector resulting from~$v$ by~flipping all its coordinates
    (thus, $v \oplus \overline{v}=1^n$). This extends naturally to~sets of~vectors: for $V \subseteq \{0,1\}^n$,
    \(\overline{V}=\{\overline{v} \colon v \in V\}.\)
    By $w(v)=\sum_{i \in [n]}v_i$, we~denote the \emph{weight} of~$v$.
    For two vectors $u,v \in \{0,1\}^n$, we~say that $u$~\emph{dominates}~$v$
    and write $u \ge v$, if $u_i \ge v_i$ for all $i \in [n]$.
    A~Boolean function $f \colon \{0,1\}^n \to \{0,1\}$ is~called
    \emph{monotone} if $f(u) \ge f(v)$, for all $u \ge v$.
    For disjoint sets $A, B \subseteq \{0,1\}^n$,
    we~say that a~function $f \colon \{0,1\}^n \to \{0,1\}$
    separates $(A,B)$ if $A \subseteq f^{-1}(1)$ and $B \subseteq f^{-1}(0)$.

   	A~\emph{Boolean circuit}~$C$ over variables $x_1,\dotsc,x_n$ is a directed acyclic graph with nodes of~in-degree zero and two. The in-degree zero nodes are labeled by variables $x_i$ and constants~$0$ or~$1$, whereas the in-degree two nodes are labeled by~binary Boolean functions. The only gate of~out-degree zero is the output of the circuit. A~circuit is called \emph{series-parallel} if there exists a labeling $\ell$ of its nodes such that for every wire $(u,v)$, $\ell(u)<\ell(v)$, and no pair of wires $(u,v), (u',v')$ satisfies $\ell(u)<\ell(u')<\ell(v)<\ell(v')$.

    A Boolean circuit~$C$ with variables $x_1,\ldots,x_n$ computes a Boolean function $f\colon\{0,1\}^n\to\{0,1\}$ in a~natural way. We define the \emph{size} of~$C$, $\size(C)$, as the number of gates in~it, and the Boolean circuit complexity of a function~$f$, $\size(f)$, as the minimum size of a~circuit computing it.
    A~circuit is~called \emph{monotone} if~all its gates compute disjunctions and conjunctions. It~is not difficult to~see that the class of~Boolean functions computed by~monotone circuits
    coincides with the class of~monotone Boolean functions. For a~monotone function~$f$, $\size_m(f)$ is~the minimum size of
    a~monotone circuit computing~$f$.

    A~sequence $(f_n)_{n=1}^{\infty}$, where $f_n \colon \{0,1\}^n \to \{0,1\}$, is~called a~\emph{function family}. Such a~family defines a~language $\cup_{n=1}^{\infty}f_n^{-1}(1)$ and we~say that the family is~\emph{explicit} if~this language is~in~\NP{}. When saying that a~language $L \subseteq \{0,1\}^*$ can be~solved by~circuits of~size $T(n)$, we~mean that it~can be~represented by a~function family $(f_n)_{n=1}^{\infty}$ where $\size(f_n) \le T(n)$, for all~$n$.

    \subsection{Arithmetic Circuits}\label{sec:ackts}
    An~\emph{arithmetic circuit} $C$ over a ring~$R$ and variables $x_1,\dotsc,x_n$ is a~directed acyclic graph with nodes
    of~in-degree zero or two. The in-degree zero nodes are labeled either by variables $x_i$ or elements of~$R$, whereas the in-degree two nodes are labeled by either $+$~or~$\times$. Every gate of out-degree zero is called an output gate.
    We will typically take $R$ to be $\mathbb{Z}$ or $\mathbb{Z}_p$ for a prime number~$p$.
    A single-output arithmetic circuit~$C$ over~$R$ computes
    a~polynomial over~$R$ in a~natural way.
	We say that $C$~computes a polynomial $P(x_1,\dotsc,x_n)$ if the two polynomials are \emph{identical} (as opposed to saying that $C$ computes $P$ if the two polynomials agree on every assignments of $(x_1,\ldots,x_n)\in R^n$).
	We define the size of~$C$ as the number of edges in~it, and the arithmetic circuit complexity of a polynomial as the minimum size of a circuit computing it.

    \subsection{SAT, MAX-SAT, OV, and Clique}
    For a~CNF formula~$F$, by~$n(F)$ and $m(F)$ we~denote the number
    of~variables and clauses of~$F$, respectively. We~write
    just~$n$ and~$m$, if~the corresponding CNF formula is~clear from the context. In~\SAT{} (\UNSAT{}), one is~given a~CNF formula and the goal is~to~check whether it~is satisfiable (unsatisfiable, respectively).
    In~$k$-\SAT{}, the given formula is in~$k$-CNF (that~is, all clauses have at~most $k$~literals).
    In \MAX{}-$k$-\SAT{}, one is~given a~$k$-CNF and an~integer~$t$
    and is~asked to~check whether it~is~possible
    to~satisfy exactly $t$~clauses.

    When designing an~algorithm for \SAT{}, one can assume that
    the input formula has a~linear (in~the number of~variables)
    number of~clauses. This is~ensured by~the following \emph{Sparsification Lemma}.
    By~$(\beta, k)$-\SAT{}, we~denote a~special case of~$k$-\SAT{} where the input $k$-CNF formula has at~most $\beta n$ clauses.

    \begin{theorem}[Sparsification Lemma, \cite{DBLP:journals/jcss/ImpagliazzoPZ01}]
    	\label{theorem:sparsification}
    	For any $k \in \mathbb{Z}_{\ge 3}$ and $\varepsilon>0$,
    	there exists $\alpha=\alpha(k, \varepsilon)$
    	and an~algorithm that, given a~$k$-CNF formula~$F$ over
    	$n$~variables, outputs $t \le 2^{\varepsilon n}$ formulas $F_1, \dotsc, F_t$ in $k$-CNF such that $n(F_i) \le n$ and $m(F_i) \le \alpha n$, for all $i \in [t]$, and $F \in \SAT$ if~and only~if $\lor_{i \in [t]}F_i \in \SAT$.
    	The running time of~the algorithm is~$O(n^{O(1)}2^{\varepsilon n})$.
    \end{theorem}

	\begin{corollary} \label{corollary:sparsification}
		There exists a~function $\beta \colon \mathbb{Z}_{\ge 3} \times \mathbb{R}_{>0} \to \mathbb{Z}_{>0}$ such that, if there exists $\varepsilon > 0$ for which $(\beta(k, \varepsilon), k)$-\SAT{} can be~solved in~time $O(2^{n / 2 + \varepsilon n})$ for~any $k \in \mathbb{Z}_{\ge 3}$, then $k$-\SAT{} can be~solved in~time $O(2^{n / 2 + 2 \varepsilon n})$.
	\end{corollary}

	\begin{proof}
		Let $\beta(k, \varepsilon) = \alpha(k, \varepsilon)$ (see \Cref{theorem:sparsification} for a~definition of~$\alpha$).
		Given an~instance of~$k$-\SAT{}, we apply the Sparsification Lemma with parameter $\varepsilon$ to~get a~sequence of $t \le 2^{\varepsilon n}$ formulas $F_1, \dotsc, F_t$, each of which is a~$(\beta(k, \varepsilon), k)$-\SAT{} instance. By~the assumption, the satisfiability of~each $F_i$ can be~checked in~time $O(2^{n / 2 + \varepsilon n})$.
		Thus, one can check the satisfiability of $\lor_{i \in [t]} F_i$ in~time $O(t \cdot 2^{n / 2 + \varepsilon n}) = O(2^{n / 2 + 2 \varepsilon n})$.
	\end{proof}

    The \emph{Strong Exponential Time Hypothesis} (\SETH{}), introduced in~\cite{DBLP:journals/jcss/ImpagliazzoPZ01,DBLP:journals/jcss/ImpagliazzoP01}, asserts that, for any $\varepsilon > 0$, there is~$k$ such that $k$-\SAT{} cannot be solved in time~$O(2^{(1 - \varepsilon) n})$. The \emph{Nondeterministic \SETH{}} (\NSETH{}), introduced in~\cite{DBLP:conf/innovations/CarmosinoGIMPS16}, extends \SETH{} by~asserting that
    \SAT{} is~difficult even for co-nondeterministic algorithms:
    for any $\varepsilon > 0$, there is~$k$ such that $k$-\SAT{} cannot be~solved in \emph{co-nondeterministic} time~$O(2^{(1 - \varepsilon) n})$. Though both these statements are stronger than $\P{}\neq\NP{}$, they are known to~be~hard to~refute:
    as~proved by~Jahanjou, Miles and Viola~\cite{DBLP:journals/iandc/JahanjouMV18},
    if~\SETH{} is~false, then there exists a~Boolean function family
    in~$\E^{\NP}$ of~series-parallel circuit size~$\omega(n)$.
    \cite{DBLP:conf/innovations/CarmosinoGIMPS16}
    noted that it~suffices to~refute \NSETH{} to~get the same circuit lower bound.

    \begin{theorem}[\cite{DBLP:journals/iandc/JahanjouMV18,DBLP:conf/innovations/CarmosinoGIMPS16}]
    	\label{theorem:local}
    	If~\NSETH{} is~false, then there exists a~Boolean function family in $\E^{\NP}$
    	of~series-parallel circuit size $\omega(n)$.
    \end{theorem}

    \SETH{}-based conditional lower bounds are known for a~wide range
    of~problems and input parameters. One of~such problems is~\problem{Orthogonal Vectors} (\OV{}): given two sets $A,B \subseteq \{0,1\}^d$ of~size~$n$, check whether there exists $a \in A$ and $b \in B$ such that $a\cdot b = \sum_{i \in [d]}a_ib_i=0$.
    It~is straightforward to~see that \OV{} can be~solved in~time $O(n^2d)$. Williams~\cite{DBLP:journals/tcs/Williams05} proved that
    under \SETH{}, there is~no algorithm solving \OV{} in~time $O(n^{2-\varepsilon}d^{O(1)})$, for any~$\varepsilon>0$.
    This follows from the following reduction.

    \begin{theorem}[\cite{DBLP:journals/tcs/Williams05}]
    	\label{theorem:satov}
    	There exists an~algorithm that, given a~CNF formula~$F$ with $n$~variables and $m$~clauses, outputs two sets $A_F,B_F \subseteq \{0,1\}^m$ such that
    	$|A_F|=|B_F|=2^{n/2}$ and $F \in \SAT{}$ if~and only~if $(A_F,B_F) \in \OV{}$.
    	The running time of~the algorithm is~$O(mn \cdot 2^{n/2})$.
    \end{theorem}

	In~a similar manner, one can reduce a~CNF formula to the $t$-$\OV$ problem involving a greater number of sets.
	Consider the $t$-$\OV$ problem, where one is given $t$ sets $A^1, \ldots, A^t \subseteq \{0, 1\}^{d}$, each of size~$n$, and the objective is to determine whether there exist elements $a^1 \in A^1, \ldots, a^t \in A^t$ such that
	\[
	\sum_{i \in [d]} a^1_i \cdot \ldots \cdot a^t_i = 0.
	\]
	\begin{lemma} \label{lemma:sat_to_t_ov}
		There exists an algorithm which, given a~CNF formula~$F$ with~$n$ variables and~$m$ clauses, along with an integer~$t$, constructs $t$ sets $A^1_F, \ldots, A^t_F \subseteq \{0, 1\}^m$ such that $|A^1| = \ldots = |A^t| = 2^{n / t}$ and $F \in \SAT$ if and only if $(A^1_F, \ldots, A^t_F) \in t$-\OV{}.
	The running time of the algorithm is~$O(m n \cdot 2^{n / t})$.
	\end{lemma}
	The proof follows the approach presented in~\cite{DBLP:journals/tcs/Williams05}.
	An instance $(A^1, \ldots, A^t) \not\in t$-$\OV$ if and only if for all vectors $a^1 \in A^1, \ldots, a^t \in A^t$, the vectors share a common coordinate with value one.
	This can equivalently be formulated by stating that the union~$Q = A^1 \cup \ldots \cup A^t$ possesses the property that for every $a^1, \ldots, a^t \in Q$, the vectors share a common one.
	Consequently, \Cref{statement:check_q_t} states that the instance $(A^1, \ldots, A^t) \not\in t$-$\OV$ if and only if there exists a circuit~$C$, composed of $\THR_{l + 1}^{lt + 1}$ gates and variables, such that $C(\overline{x}) = 0$ for every $x \in A^1 \cup \ldots \cup A^t$.

	A~similar reduction allows to~solve \MAX{}-3-\SAT{}
	by~finding a~$4$-clique in a~$3$-uniform hypergraph.
	A~subset of~$l$~nodes in a~$k$-hypergraph is~called an~\emph{$l$-clique},
	if any $k$~of them form an~edge in~the graph.

	\begin{theorem}[\cite{williams2007algorithms,DBLP:conf/soda/LincolnWW18}]
		\label{theorem:max_3_sat_4_clique}
		There exists an~algorithm that, given a~$3$-CNF formula~$F$
		with $n$~variables and an~integer~$t$, outputs a~$4$-partite $3$-uniform hypergraph~$G$ with parts of~size $k = n^{O(1)} 2^{n/4}$ such that $G$~has a~$4$-clique if~and only~if it~is possible
		to~satisfy exactly $t$~clauses of~$F$.
	\end{theorem}
	\begin{proof}
		To~construct the graph~$G$, partition
		the~variables of~$F$ into four groups $A_0, A_1, A_2, A_3$
		of~size $n/4$. Then, label each clause of~$F$ by some $i \in \{0,1,2,3\}$ such that the clause does not contain variables
		from~$A_i$. Then, for any $i \in \{0,1,2,3\}$,
		assigning variables from all parts except from~$A_i$,
		determines how many clauses labeled~$i$ are satisfied.
		Define tensors~$\mathcal T_0, \mathcal T_1, \mathcal T_2, \mathcal T_3 \in \mathbb{Z}_{\ge 0}^{2^\frac{n}{4} \times 2^\frac{n}{4} \times 2^\frac{n}{4}}$ as~follows: for $U_1, U_2, U_3 \in \{0,1\}^{\frac n4}$, let $\mathcal T_i[U_1, U_2, U_3]$ be~equal to~the number of~clauses with label~$i$ when
		the variables from groups
		$A_{(i+1) \bmod 4}$,
		$A_{(i+2) \bmod 4}$, and
		$A_{(i+3) \bmod 4}$
		are assigned as~in $U_1$, $U_2$, and $U_3$.
		Then, it~is possible to~satisfy $t$~clauses in~$F$ if~and
		only~if there exist $U_0, U_1, U_2, U_2 \in \{0,1\}^{\frac n4}$ such that
		\[
		\mathcal T_0[U_1,U_2,U_3] + \mathcal T_1[U_2,U_3,U_0] +\mathcal  T_2[U_3,U_0,U_1] +\mathcal  T_3[U_0,U_1,U_2] = t.
		\]
		As~$F$ contains at~most $O(n^3)$ clauses, $t=O(n^3)$.
		This makes it~possible to~enumerate in~polynomial time
		all values of~the four terms above. For an~integer~$q$ and $i \in \{0,1,2,3\}$, define a~tensor $\mathcal{A}_{i,q} \in \{0,1\}^{2^\frac{n}{4} \times 2^\frac{n}{4} \times 2^\frac{n}{4}}$~as
		\[
		\mathcal A_{i,q}[U_1,U_2,U_3] = [T_i[U_1,U_2,U_3] = q].
		\]
		Then, one can satisfy $t$~clauses in~$F$ if~and only~if
				\[
		\sum_{\substack{U_0, U_1, U_2, U_3 \in \{0,1\}^{\frac n4}\\t_0+t_1+t_2+t_3=t}} \mathcal A_{0,t_0}[U_1,U_2,U_3] \cdot \mathcal A_{1,t_1}[U_2,U_3,U_0] \cdot \mathcal A_{2,t_2}[U_3,U_0,U_1] \cdot \mathcal A_{3,t_3}[U_0,U_1,U_2] > 0.
		\]

		For fixed $t_0, t_1, t_2, t_3$, such that $t_0 + t_1 + t_2 + t_3 = t$,
        the tensors $\mathcal{A}_{0, t_0}, \mathcal{A}_{1, t_1}, \mathcal{A}_{2, t_2}, \mathcal{A}_{3, t_3}$ may be~viewed as a~description of~edges
        of~a~4-partite 3-uniform hypergraph~$G_{t_0, t_1, t_2, t_3}$.
		There is a~4-clique in $G_{t_0, t_1, t_2, t_3}$ if~and only~if one can satisfy
        $t_0$~clauses with~label~$0$, $t_1$~clauses with label~$1$, and so on.
		Let $G$~be a~superimposition of~all such graphs. Then, $G$~contains a~4-clique if and only~if one can satisfy exactly  $t$~clauses of~$F$.
	\end{proof}

	\subsection{Rigidity and Tensor Rank}
	\label{sec:matrices_and_tensors}
	For a~field $\mathbb{F}$, by $\mathbb{F}^{a \times b}$ we denote the set of all matrices of size $a \times b$ over~$\mathbb{ F }$.
	Similarly, by $\mathbb{F}^{a \times b \times c}$ we denote the set of all three-dimensional tensors of shape $a \times b \times c$ over~$\mathbb{F}$.
	For two tensors $A$~and~$B$ (of arbitrary shape), by $A \otimes B$ we denote a~tensor product of $A$~and~$B$.
	For a~matrix $M \in  \mathbb{F}^{a \times b}$, by  $|M|$ we denote the number of nonzero entries of~$M$.

	For~a~matrix $M \in \mathbb{F}^{a \times b}$, we say that it has $r$-rigidity~$s$ if~it is necessary to~change at~least $s$ entries of~$M$ to~reduce its rank to~$r$.
	That~is, for~each decomposition $M = R + S$ such~that $\rank(R) \le r$, it~holds that $|S| \ge s$.

	The rank of~a~three-dimensional tensor is a~natural extension
	of~the matrix rank.
	For a~tensor $\mathcal{A} \in \mathbb{F}^{n \times n \times n}$, we define its rank, $\rank(\mathcal{A})$, as the smallest integer~$r$ such that there exist~$r$ tuples of vectors $a_l, b_l, c_l \in \mathbb{F}^{n}$ for which
	\[
	  \mathcal{A} = \sum_{l \in [r]} a_l \otimes b_l \otimes c_l,
	\]
	or equivalently,
	\[
		\mathcal{A}[i, j, k] = \sum_{l \in [r]} a_l[i] b_l[j] c_l[k],
	\]
	for all $i, j, k \in [n]$.

	We denote by $\omega$ the smallest real number such that any two $n \times n$ matrices can be multiplied in time $O(n^{\omega + \varepsilon})$ for any $\varepsilon > 0$ using only field operations\footnote{The value of $\omega$ may depend on the field over which the calculations are performed~\cite{DBLP:journals/siamcomp/Schonhage81}.}.

	Consider the three-dimensional tensor $\mathcal{A}_n \in \mathbb{F}^{n^2 \times n^2 \times n^2}$:
	\[
		\mathcal{A}_n[(i, j), (i, k), (k, j)] = 1,
	\]
	for all $i, j, k \in [n]$, with all other entries being zero.
	Using an approach based on the work of~Strassen~\cite{strassen1969gaussian}, for any positive integer~$k$, if $\rank(\mathcal{A}_k) = q$, then one can construct an arithmetic circuit of size~$O(n^{\log_k(q)})$ to perform multiplication of two $n \times n$ matrices.
	Thus, $\omega$ satisfies the following equation:
	\[
	   \omega = \inf_{k \in \mathbb{Z}_{>0}} \log_k \rank(\mathcal{A}_k).
	\]
	In other words, for sufficiently large~$k$, we have that $\rank(\mathcal{A}_k) \geq k^{\omega}$.
	Specifically, if $n = k^2$, then $\mathcal{A}_k \in M_{n \times n \times n}$, and $\rank(\mathcal{A}_k) \geq n^{\omega / 2}$.
	Therefore, if $\omega > 2$, this yields superlinear lower bounds on arithmetic circuits and on the rank of the multiplication tensor.

	The best known upper bound on $\omega$ is $2.371339$~\cite{DBLP:conf/soda/AlmanDWXXZ25}.
	Further details
	on~the matrix multiplication tensor
	can be~found in~\cite{DBLP:journals/theoretics/AlmanW24,DBLP:journals/toc/Blaser13}.

    \section{Circuit Lower Bounds}
    \label{sec:circuit_lower_bounds}
	\subsection{Boolean Functions of High Monotone Circuit Size}
    In~this section, we~prove \Cref{theorem:monotone_conp}.

    \thmmonotone*

	Combining this with \Cref{corollary:monotone_nseth} and Theorem~\ref{theorem:local}, we~get
	win-win circuit lower bounds. Proving any of~these two
    circuit lower bounds is a~challenging open problem.

	\crlMonotone*

    For the proof of~\Cref{theorem:monotone_conp}, we~need two technical lemmas.
    Recall that the reduction from \SAT{} to \OV{} (see Theorem~\ref{theorem:satov}),
    given a~formula~$F$, produces two sets $A_F, B_F \subseteq \{0,1\}^{m(F)}$
    such that $|A_F|=|B_F|=2^{n(F)/2}$ and $F \in \SAT{}$ if~and
    only~if $(A_F, B_F) \in \OV{}$.

    \begin{lemma}
        \label{lemma:monotoneseparation}
        Let $F$~be a~CNF formula.
        Then, $F \not \in \SAT{}$ if~and only~if
        $(A_F, \overline{B_F})$ can be~separated
        by~a~monotone function.
    \end{lemma}

    For a~CNF formula~$F$, let $f_F \colon \{0,1\}^{m(F)} \to \{0,1\}$
    be~defined as~follows:
    \begin{equation}
        \label{eq:ff}
        f_F(x)=[\forall b \in \overline{B_F} \colon b \not \ge x].
    \end{equation}
    It~is immediate that $f$~is monotone and that $\overline{B_F} \subseteq f_F^{-1}(0)$.

    \begin{proof}[Proof of~Lemma~\ref{lemma:monotoneseparation}]
        Assume that $F \not \in \SAT{}$.
        We~show that the function~$f_F$ separates $(A_F, \overline{B_F})$.
        To~do this, it~suffices to~show that $A_F \subseteq f_F^{-1}(1)$.
        If~this is~not the case, then there is~$a \in A_F$ such that $f_F(a)=0$, that~is,
        there exists $b \in \overline{B_F}$ such that $b \ge a$.
        Hence, there is~no $i \in [m]$ such that $b_i=0$ and $a_i=1$.
        In~turn, this means that, for $\overline{b} \in B_F$, there is~no $i \in [m]$ such that $\overline{b}_i=1$ and $a_i=1$, meaning that $(A_F, B_F) \in \OV{}$, contradicting
        $F \not \in \SAT{}$.

        For the reverse direction, assume that for some monotone function $h \colon \{0,1\}^m \to \{0,1\}$, it~holds that $A_F \subseteq h^{-1}(1)$ and $\overline{B_F} \subseteq h^{-1}(0)$. By~the monotonicity of~$h$, for every $a \in A_F$ and every $b \in \overline{B_F}$, $b \not \ge a$. Hence, for every $a \in A_F$ and every $b \in \overline{B_F}$, there exists $i \in [m]$ such that $b_i=0$ and $a_i=1$. Switching from
        $\overline{B_F}$ to~$B_F$, we~get that, for every $a \in A_F$ and every $b \in B_F$, there exists $i \in [m]$ such that $b_i=1$ and $a_i=1$, meaning that $(A_F, B_F) \not \in \OV{}$ and $F \not \in \SAT{}$.
    \end{proof}

    By $\mathcal{F}_{n,k, \beta}$ denote the set of $k$-CNF formulas with
    $n$~variables and at~most $\beta n$ clauses.
    Any formula $F \in \mathcal{F}_{n,k,\beta}$ can be~encoded in~binary using at~most $\gamma n\log n$
    bits (for some $\gamma=\gamma(k, \beta)$): each variable is encoded using $\log n$
    bits, all other symbols (parentheses as~well~as negations, disjunctions, and conjunctions require a~constant number of~bits). Hence,
    \[|\mathcal{F}_{n,k,\beta}| \le 2^{\gamma n \log n}.\]
    Let us call $\mathcal{W}_{n, k}$ the set of all binary strings of length $n$ and weight $k$:
    \[\mathcal{W}_{n, k} = \{x \in \{0, 1\}^n \colon w(x) = k\}.\]
	\begin{lemma}
		\label{lemma:encoding_numbers}
		There exists an~injective encoding~$e \colon \{ 0, 1 \}^{\leq n} \to \mathcal{W}_{4n, 2n}$ such that computing and~inverting~$e$ takes time linear in~$n$.
	\end{lemma}

	\begin{proof}
		For $x \in \{0,1\}^{\le n}$, let $e(x)$ as~follows be~any balanced string of~length $4n$ that starts with $1^{|x|}0$ followed by~$x$.
		This~encoding is~injective, as~the~initial segment~$1^{|x|} 0$ allows us~to~uniquely identify~$x$ within~$e(x)$.
		Both~computing and~inverting~$e$ can~be~accomplished in~linear time with~respect to~$n$.
		Indeed, computing~$e(x)$ is~straightforward.
		Conversely, given~$e(x)$, one~can first extract the~size of~$x$ and~then decode~$x$.
	\end{proof}

	As~a~simple corollary, applying~\Cref{lemma:encoding_numbers} to~Boolean formulas, one~can obtain the~following result.
    \begin{lemma}
        \label{lemma:encoding}
        There exists a~parameter $l=O(n\log n)$ and an~injective encoding
        \[e \colon \mathcal{F}_{n,k,\beta} \to \mathcal W_{l, l/2} \]
		such that computing and inverting~$e$ takes time polynomial in~$n$.
    \end{lemma}
    \begin{proof}
        Given a~formula $F \in \mathcal{F}_{n, k, \beta}$, we~can encode~it in~binary using at~most $\gamma n \log n$~bits (as~stated above).
		Then, applying~\Cref{lemma:encoding_numbers}, we~obtain the~desired result.
    \end{proof}

	\begin{proof}[Proof of~Theorem~\ref{theorem:monotone_conp}]
        Assuming that \coNP{} has small monotone circuits, we~design
        a~fast co-nondeterministic algorithm for \SAT{}.
		Assume that all monotone functions over $N$~variables in~\coNP{} can be computed
		by~monotone circuits of size $2^{o(N/\log N)}$.

		Let $F$~be a~$k$-CNF over $n$~variables.
		Thanks to~Corollary~\ref{corollary:sparsification}, we~may assume that~$m(F) \leq \beta n$, where~$\beta = \beta(k, \frac{\varepsilon}{2})$.
		Then, solving~$\mathcal{F}_{n, k, \beta}$ in~co-nondeterministic time~$O(2^{n / 2 + \varepsilon' n})$ for~$\varepsilon' = \frac{\varepsilon}{2}$ will~be~sufficient for~a~contradiction.

        Below, we~define a~universal function~$f$ containing $f_F$, for all $F \in \mathcal{F}_{n,k, \beta}$,
        as~subfunctions
        (recall the definition~\eqref{eq:ff} of~$f_F$).
        Let $N=l+m$, where $l=O(n\log n)$ and $e$ is the function from Lemma~\ref{lemma:encoding} and $m=\beta n$ (hence, $N=O(n\log n)$).
        We define $f \colon \{0,1\}^N \to \{0,1\}$
        as~follows. Let $(c,x)$, where $c \in \{0,1\}^l$ and $x \in \{0,1\}^m$,
        be~an~input of~$f$. Then,
        \[
            f(c, x)=
            \begin{cases}
                1, & \text{if $w(c)>l/2$,}\\
                0, & \text{if $w(c)<l/2$,}\\
                1, & \text{if $w(c)=l/2$ and $e^{-1}(c)=\varnothing$,}\\
                0, & \text{if $w(c)=l/2$ and $e^{-1}(c) \in \SAT{}$,}\\
                f_F(x), & \text{if $w(c)=l/2$ and $F=e^{-1}(c) \not \in \SAT{}$.}
            \end{cases}
        \]
        We~now ensure three important properties of~$f$.
        \begin{enumerate}
            \item \emph{The function~$f$ is~monotone.} For the sake of~contradiction,
            assume that $(c,x) \ge (c',x')$, but $0=f(c,x)<f(c',x')=1$.
            Since $f(c',x')=1$, $w(c') \ge l/2$.
            If $w(c')>l/2$, then $f(c,x)=1$, since $w(c)\ge w(c')>l/2$.
            Hence, assume that~$w(c')=w(c)=l/2$.
            Since~$c\ge c'$,
            we~conclude that $c=c'$ (implying that~$e^{-1}(c) \neq \varnothing$).
            If~$e^{-1}(c) \in \SAT{}$, then~$f(c', x')=0$.
            Hence~$e^{-1}(c) \not \in \SAT{}$, thus~$f(c,x)=f_F(x)$ and~$f(c',x')=f_F(x')$,
            where $F=e^{-1}(c)=e^{-1}(c')$, and $f_F$ is~clearly monotone.

            \item \emph{The function~$f$ is~explicit: $f \in \coNP$.} Assume that $f(c,x)=0$.
            This can happen for one of~the three following reasons.
            \begin{itemize}
                \item $w(c)<l/2$. This is~easily verifiable.
                \item $w(c)=l/2$ and $e^{-1}(c) \in \SAT{}$. To~verify this,
                one computes $F=e^{-1}(c)$ (in~time $O(\poly(n))$, due to~Lemma~\ref{lemma:encoding}), guesses its satisfying
                assignment and verifies~it.
                \item $w(c)=l/2$, $F=e^{-1}(c) \not \in \SAT{}$, and $f_F(x)=0$. Recall that if $F \in \SAT{}$, then we can
                guess and verify its satisfying assignment, so~in this case
                there~is no~need to~verify the unsatisfiability of~$F$.
                Since $f_F(x)=0$, there exists $b \in \overline{B_F}$ such that $b \ge x$.
                To~verify this, one computes $F=e^{-1}(c)$, guesses the corresponding assignment to~the second
                half of~the variables of~$F$, ensures that it~produces the vector $\overline{b}$, and verifies that $b \ge x$.
            \end{itemize}
            \item
            Assume that $f$~can be~computed by a~monotone circuit~$C$ of~size $2^{o(N/\log N)}$.
			We~show that, then, for~any~$k \in \mathbb{Z}_{\geq 3}$,
			$k$-\SAT{} can~be~solved in~input-oblivious co-nondeterministic time~$O(2^{n / 2 + \varepsilon' n})$.
            Since $N=O(n\log n)$, then $\operatorname{size}(C)=O(2^{\varepsilon' n})$.

            Let $F \in \mathcal{F}_{n,k,\beta}$ be an~unsatisfiable formula with $n$~variables
            and no~more than $\beta n$ clauses. The following algorithm verifies its unsatisfiability in~input-oblivious nondeterministic time $O(2^{n / 2 + \varepsilon' n})$.
            \begin{quote}
                \begin{enumerate}
                    \item In~time $O(n^22^{n/2})$, generate the sets $A_F$~and~$B_F$.
                    \item Guess a~monotone circuit~$C$.
                    This can be~done in~time $O(\operatorname{size}(C))=O(2^{\varepsilon' n})$.
                    \item Compute $c=e(F)$ and substitute the first $l$~variable of~$C$
                    by~$c$. Call the resulting circuit~$C_F$.
                    \item Verify that $C_F$ separates $(A_F, \overline{B_F})$. To~do this, check that $C_F(a)=1$, for every $a \in A_F$. Then, check that
                    $C_F(b)=0$, for every $b \in \overline{B_F}$. The running time
                    of~this step~is
                    \[O((|A_F|+|\overline{B_F}|) \cdot \operatorname{size}(C_F))=O(2^{n/2} \cdot 2^{\varepsilon' n})=O(2^{n / 2 + \varepsilon' n}).\]
                    \item If $C_F$ is~monotone and separates $(A_F, \overline{B_F})$, then
                    we~are certain that $F \not \in \SAT{}$, thanks to~Lemma~\ref{lemma:monotoneseparation}.
                \end{enumerate}
            \end{quote}
        \end{enumerate}
    \end{proof}

	\subsection{Functions of High Threshold Size}
	In~this section, we aim to derive circuit lower bounds under the assumption of~$\NETH$.
	It is important to observe that the preceding reduction is insufficient for solving $3$-\SAT{} in time $2^{o(n)}$, as the reduction to the $\OV$ problem operates in time~$O(2^{n/2})$.
	However, by reducing $3$-\SAT{} to~$t$-$\OV$ for some $t = \omega(1)$, the reduction can be executed in time~$2^{o(n)}$.
	\thmthreshold*
	\begin{proof}
	First, we~demonstrate how small~$\THR_{l + 1}^{lt + 1}$ circuits can be~used to~verify whether a~given formula is~in $\UNSAT$.
	Subsequently, we~show how to construct an~explicit function based on this observation.

	Given an arbitrary $\varepsilon > 0$ and a $3$-\SAT{} instance~$F$, apply~\Cref{corollary:sparsification} to ensure that the number of clauses~$m(F)$ satisfies $m(F) = \beta n$, where~$\beta = \beta(3, \varepsilon)$.
	Then, add artificial variables to make $n$ a multiple of $t$.
	Next, apply~\Cref{lemma:sat_to_t_ov} to construct the sets $A^1, \ldots, A^t$ such that $(A^1, \ldots, A^t) \not\in t$-$\OV$ if and only if $F \in \UNSAT$.

	Let $f_F  \colon \{ 0, 1 \}^{\beta n} \to \{ 0, 1 \}$ be a function such that $f_{F}^{-1}(0) = \cup_{i = 1}^{t} \overline{A^i}$.
	Then, $f_{F} \in Q^{n}_{t}$ if and only if $(A^{1}, \ldots, A^{t}) \not \in t$-$\OV$.
	Subsequently, nondeterministically guess a circuit~$C$ of size~$2^{o(n)}$, composed only of $\THR_{l + 1}^{lt + 1}$ gates for arbitrary~$l \ge 1$.
	Then verify that $C(\overline{x}) = 0$ for every $x \in \cup_{i = 1}^{t} \overline{A^i}$.
	The total verification time is~$O(t \cdot \size(C) \cdot 2^{n/t}) = 2^{o(n)}$.
	\Cref{statement:check_q_t} guarantees that if $F \in \UNSAT$, then such a circuit~$C$ must exist.
	Therefore, if no such circuit~$C$ of size~$2^{o(n)}$ is found, it must be the case that its size is~$2^{\Omega(n)}$ or $F \in \SAT$.

	In~order to construct an explicit function from $\coNP_{/\poly}$, consider the sequence of 3-CNF formulas $F_1, F_2, \ldots \in \UNSAT$, such that the $i$-th formula contains exactly $i$ clauses.
	Moreover, for each $i$, we select a formula $F_i$ such that the corresponding function $f_{F_i}$ has the maximal size of $\THR_{l + 1}^{lt + 1}$ circuits among all functions $f_{F'}$ induced by 3-CNF formulas $F'$ with $i$ clauses.
	Then, this sequence of functions achieves a $2^{\Omega(n)}$ lower bound against $\THR_{l + 1}^{lt + 1}$ circuits under $\NETH$.
	Now, we construct an algorithm from $\coNP_{/\poly}$ which, on input $x$, outputs $f_{|x|}(x)$.
	The advice for our algorithm on inputs of length $i$ is the formula $F_i$.

	To compute $f_{F_{i}}(x)$, where $i = |x|$, we nondeterministically guess an index $k \in [t]$ and an assignment to the variables with indices from $1 + (k - 1) \cdot \frac{n(F_i)}{t}$ to $k \cdot \frac{n(F_i)}{t}$, and from this we obtain the vector of satisfying clauses $y \in A^{k}$ for the formula $F_{i}$.
	If $\bar{y} \ge x$, we output $0$; otherwise, we output $1$.
	It is evident that $f_{F_{i}}(x) = 0$ if and only if such an index $i$ and assignment exist.
	\end{proof}

	\section{Matrices of High Rigidity}
	\label{sec:rigid_matrices}
	In~this section, we~show that low rigidity matrices can be~utilized to~solve \MAX{}-3-\SAT{} more efficiently.
	Thus, assuming that \MAX{}-3-\SAT{} cannot be~solved in~co-nondeterministic time $O(2^{(1 - \varepsilon) n})$, for any~$\varepsilon > 0$,
	we~obtain a~generator of~high rigidity matrices. Throughout this section,
    rigidity decompositions are considered over a~finite field~$\mathbb{F}$.

	\todo[inline]{ reviewer: I am less convinced by Theorem 4. By standard arguments from pseudo-randomness, it follows that even $2^{\varepsilon n}$ non-deterministic circuit size lower bounds for E imply constructions of rigid matrices stronger than the conclusion in Theorem 4. It's true that the assumption in this PRG-based construction is non-uniform, but it is for E, while the assumption in Theorem 4 is for Max-k-SAT. Moreover, as the authors themselves point out, there are no known lower bound consequences of non-trivial non-deterministic algorithms for Max-k-SAT, so there is no win-win conclusion to be had here, at least for now.
	}

	\thmHighRigidity*

    \begin{proof}
    	Take a~$3$-CNF formula over $n$~variables and an~integer~$t$ and transform~it, using \Cref{theorem:max_3_sat_4_clique}, into a~$4$-partite $3$-uniform hypergraph with parts $H_0, H_1, H_2, H_3$
    	of~size $k = n^{O(1)} 2^{\frac{n}{4}}$. Recall that $G$~contains a~$4$-clique if~and only~if one can satisfy $t$~clauses of~$F$.

        Let~$\mathcal A_0, \mathcal A_1, \mathcal A_2, \mathcal A_3 \in \{0, 1\}^{k \times k \times k}$ be three-dimensional tensors encoding the edges of~$G$. $\mathcal A_i$ is~responsible for storing edges spanning nodes from all parts except for~$H_i$: for example,
        $G$~has an~edge~$(u_1, u_2, u_3)$,
        where $u_1 \in H_1$, $u_2 \in H_2$, and $u_3 \in H_3$,
		if and only if~$\mathcal A_0[u_1, u_2, u_3] = 1$.
		Let
        \[
        R=\sum_{j_0, j_1, j_2, j_3 \in [k]} \mathcal  A_{0}[j_1,j_2,j_3] \cdot \mathcal A_{1}[j_2,j_3,j_0] \cdot \mathcal A_{2}[j_3,j_0,j_1] \cdot \mathcal A_{3}[j_0,j_1,j_2].
        \]
        Then, $G$~has a~$4$-clique if~and only~if $R>0$.
        Also, for $j_0, j_1 \in [k]$, let
        \[
            R_{j_0, j_1} = \sum_{j_2, j_3 \in [k]}  \mathcal A_{0}[j_1,j_2,j_3] \cdot \mathcal A_{1}[j_2,j_3,j_0] \cdot \mathcal A_{2}[j_3,j_0,j_1] \cdot \mathcal A_{3}[j_0,j_1,j_2].
        \]
        Thus,
        \[
			R = \sum_{j_0, j_1 \in [k]} R_{j_0, j_1}.
        \]
        Now, for fixed $j_0, j_1 \in [k]$, define vectors $v, u \in \{0,1\}^{k}$ as~follows: $v[i] = \mathcal A_3[j_0, j_1, i]$ and~$u[l] = \mathcal A_2[l, j_0, j_1]$.
        Also, define two~matrices $M, L \in \{0, 1\}^{k \times k}$: $M[i, l] = \mathcal A_0[j_1, i, l]$ and~$L[i, l] = \mathcal A_1[i, l, j_0]$.
        Hence,
        \[
			R_{j_0, j_1} = \sum_{j_2, j_3 \in  [k]} M[j_2, j_3] \cdot L[j_2, j_3] \cdot u[j_3] \cdot v[j_2].
        \]

        Now, assume that~$M$ and~$L$ have $r$-rigidity~$s$, that~is,
        \begin{align*}
            M &= R_M + S, \\
            L &= R_L + T,
        \end{align*}
        where $R_M, R_L, S, T \in  \mathbb{F}^{k \times k}, \; \rank(R_M), \rank(R_L) \leq r$ and~$|S|, |T| \leq s$.

        Since $\rank(R_M) \leq r, \ \rank(R_L) \leq r$, it follows that there exist vectors: $a_{i, M}, a_{i, L} \in \mathbb{F}^{k}$ and~$b_{i, M}, b_{i, L} \in \mathbb{F}^{k}$ such that
        \begin{align*}
			M &= \left(\sum_{i \in  [r]} a_{i, M} \cdot b^T_{i, M}\right) + S, \\
			L &= \left(\sum_{i \in [r]} a_{i, L} \cdot b^T_{i, L}\right) + T.
        \end{align*}
        We guess these vectors for the matrices~$M$ and~$L$.
        Additionally, we guess only nonzero entries of~$S$ and~$T$, since they are~sparse.

        For fixed~$j_0$, the~time complexity to~guess and~verify the~decomposition of~$L$ is~$O(r k^2 + s + k^2)$, which includes the~time to~multiply the~vectors, sum~them, add~$S$, and~verify that the~result equals~$L$.
        The same applies when fixing~$j_1$.
        Thus, the overall time complexity is~$O(r k^3 + ks + k^3) = O(r k^3 + ks)$, for all $j_0, j_1$.

        Hence,
        \begin{align*}
			R_{j_0, j_1} &= \sum_{j_2 \in [k]} \sum_{j_3 \in  [k]} \left(S + \sum_{i \in [r]} a_{i, M} \cdot b^T_{i, M}\right)[j_2, j_3] \cdot \left(T + \sum_{l \in [r]} a_{l, L} \cdot b^T_{l, L}\right)[j_2, j_3] \cdot u[j_3] \cdot v[j_2] = \\
						 &= \sum_{j_2 \in [k]} \sum_{j_3 \in [k]} \left(\sum_{i \in [r]} a_{i, M} \cdot b^T_{i, M}\right)[j_2, j_3] \cdot \left(\sum_{l \in [r]} a_{l, L} \cdot b^T_{l, L}\right)[j_2, j_3] \cdot u[j_3] \cdot v[j_2] + \\
						 &+ \underbrace{\sum_{j_2 \in [k]} \sum_{j_3 \in [k]} u[j_3] \cdot v[j_2] \cdot \left(S[j_2, j_3] \cdot T[j_2, j_3] + S[j_2, j_3] \cdot R_L[j_2, j_3] + R_M[j_2, j_3] \cdot T[j_2, j_3]\right)}_{R_{j_0, j_1}'}.
        \end{align*}
        Since $S$ and~$T$ are sparse, we can compute the $R_{j_0, j_1}'$ in time~$O(s)$, for fixed~$j_0, j_1$.
        Now, it~remains to~evaluate
        \begin{align*}
			R_{j_0, j_1} - R_{j_0, j_1}'=&\sum_{j_2 \in [k]} \sum_{j_3 \in [k]} \left(\sum_{i \in [r]} a_{i, M} \cdot b^T_{i, M}\right)[j_2, j_3] \cdot \left(\sum_{l \in [r]} a_{l, L} \cdot b^T_{l, L}\right)[j_2, j_3] \cdot u[j_3] \cdot v[j_2] = \\
										 &\sum_{j_2 \in [k]} \sum_{j_3 \in [k]} \left(\sum_{i \in [r]} a_{i, M}[j_2] \cdot b_{i, M}[j_3]\right) \cdot \left(\sum_{l \in [r]} a_{l, L}[j_2] \cdot b_{l, L}[j_3]\right)\cdot u[j_3] \cdot v[j_2] = \\
										 &\sum_{i \in [r]} \sum_{l \in [r]} \sum_{j_2 \in [k]} \sum_{j_3 \in [k]} a_{i, M}[j_2] \cdot b_{i, M}[j_3] \cdot  a_{l, L}[j_2] \cdot b_{l, L}[j_3] \cdot u[j_3] \cdot v[j_2] = \\
										 &\sum_{i \in [r]} \sum_{l \in [r]} \left(\sum_{j_2 \in [k]} a_{i, M}[j_2] \cdot a_{l, L}[j_2] \cdot v[j_2]\right) \left( \sum_{j_3 \in [k]} b_{i, M}[j_3] \cdot b_{l, L}[j_3] \cdot u[j_3]  \right).
        \end{align*}
        This sum can be computed in time~$O(r^2 k)$.
        Thus, the~total running~time (for all~$j_0,j_1$) is~$O(r k^3 + ks + k^2 s + r^2 k^3) = O(k^2 s + r^2 k^3)$.
        If~$r = k^{\frac{1}{2} - \delta}$ and~$s = k^{2 - \delta}$ for some~$\delta > 0$, the~running time becomes~$O(k^{4 - \delta})$.

		We~can construct a generator~$g$ that takes the~following inputs: the encoding~of a~$3$-CNF formula, an~integer~$t$, $j_0, j_1$, and a~$0/1$ bit. If this bit is~$0$, we~output~$M$; otherwise, we~output~$L$.
		If~the input to~the generator is not valid, it outputs an~empty family.
        Therefore, if~\MAX{}-$3$-\SAT{} cannot be~solved in co-nondeterministic time~$O(2^{(1 - \varepsilon)n})$ for any~$\varepsilon > 0$, then for~infinitely many $n$, the~generator outputs at~least one~matrix with $k^{\frac{1}{2} - \delta}$-rigidity~$k^{2 - \delta}$, for any~$\delta > 0$, and~$k = n^{O(1)} 2^{\frac{n}{4}}$.
        The generator takes~$n^{O(1)} = \log^{O(1)}(k)$ input bits, works in time polynomial in~$k$, and outputs a~matrix of dimensions~$k \times k$.
    \end{proof}

	As~a~simple corollary, one might weaken the assumption and obtain weaker matrix rigidity, while still improving the known explicit construction by~\cite{DBLP:journals/ipl/ShokrollahiSS97}.
	Recall that~\cite{DBLP:journals/ipl/ShokrollahiSS97} constructed matrices with $r$-rigidity $\frac{n^2}{r} \log (n / r)$.

	\begin{corollary} \label{corollary:weaken_rigidity}
		If \MAX{}-3-\SAT{} cannot be~solved in~co-nondeterministic time $O(2^{0.92 n})$, then there exists a~generator $g \colon \{0, 1\}^{\log^{O(1)} k} \to \mathbb{F}^{k \times k}$ computable in~time polynomial in~$k$ such that, for infinitely many~$k$, there exists a~seed~$s$ for which $g(s)$ has $k^{0.34}$-rigidity~$k^{1.68}$.
	\end{corollary}
	\begin{proof}
		The generator is the one constructed in~the proof of~\Cref{thm:high_rigidity}.
		Substituting $r = k^{0.34}$ and $s = k^{1.68}$ into the running time $O(k^2 s + r^2 k^3)$ yields the running time $O(k^{3.68}) = O(2^{0.92n})$, thereby completing the proof.
		Note that $r \cdot s = k^{2.02}$, so further weakening the assumption will result in worse rigidity than that achieved in~\cite{DBLP:journals/ipl/ShokrollahiSS97}.
	\end{proof}

	\section{Tensors of~High Rank} \label{sec:high_rank_tensors}
    In~this section, we~show how to~generate high rank tensors under an~assumption
    that \MAX{}-3-\SAT{} is~hard.
    Throughout this section, we~assume that $\mathbb{F}$ is a~finite field over which
    rigidity and tensor decompositions are considered.
    We~start by~proving two auxiliary lemmas. The first one
    shows how rectangular matrix multiplication can be reduced to square matrix multiplication.
    \begin{lemma} \label{lemma:matrix_multiplication}
        For $a, b \ge n$, the~product of two matrices $A \in \mathbb F^{a \times n}$ and $B \in \mathbb{F}^{n \times b}$ can be~computed in~time~$O(abn^{\omega - 2})$.
    \end{lemma}
    \begin{proof}
        Partition $A$~and~$B$ into $n \times n$-matrices $A_1, \dotsc, A_{a/n}$ and $B_1, \dotsc, B_{b/n}$.
        Then, to~compute $A \cdot B$, it~suffices to~compute $A_i \cdot B_j$,
        for all $i \in [a/n]$ and $j \in [b/n]$. The resulting running time is
        \[O\left( \frac{a}{n}\cdot \frac{b}{n}\cdot n^{\omega}\right)=O(abn^{\omega-2}).\]
    \end{proof}
	The next lemma shows how one can compute the value of a~specific
    function for all inputs using fast matrix multiplication algorithms.
    \begin{lemma} \label{lemma:evaluate_sum}
        Let $q \ge k$ and $A, B, C \in \mathbb{F}^{q \times k}$. Let also
        $f \colon [k]^3 \to \mathbb{F}$ be defined as~follows:
        \[
			f(i, j, m) = \sum_{l \in [q]} A[l, i] B[l, j] C[l, m].
        \]
        One can compute $f(i,j,m)$, for all $(i, j, m) \in [k]^3$ simultaneously,
        in~time $O(q k^{\omega})$.
    \end{lemma}

    \begin{proof}
        Fix $i \in [k]$ and let $D \in \mathbb{F}^{q \times k}$ be~defined by
        $D[l, j] = B[l, j] \cdot A[l, i]$. Then,
        \[
			f(i, j, m) = \sum_{l \in [q]} D[l, j] C[l, m] = (D^T \cdot C)[j, m].
        \]
		One may partition~$D^{T}$ and~$C$ into $k \times k$~matrices, denoted as~$D_1, \ldots, D_{q / k}$ and~$C_1, \ldots, C_{q / k}$, respectively.
		The~product $D^{T} \cdot C$ can thus be expressed as the summation $\sum_{i \in [q / k]} D_i \cdot C_i$.
		The~computation of a~single product~$D_{i} \cdot C_i$ can be~performed in~$O(k^{\omega})$ time.
		Therefore, the~total computation time for~$D^{T} \cdot C$ is~$O(q k^{\omega - 1})$.
		Summing over all~$i \in [k]$, we obtain the~desired upper bound.
    \end{proof}

	Now, we~are ready to~prove the main result of~this section.
	\begin{theorem} \label{thm:high_rank}
        Let $k=n^{O(1)}2^{n/4}$ and $r \ge \sqrt k$, $q \ge k$.
		There exist functions $g_1  \colon \{0, 1\}^{\log^{O(1)} k} \to  \mathbb{F}^{k \times k}$ and $g_2  \colon \{0, 1\}^{\log^{O(1)} k} \to \mathbb{F}^{k \times k \times k}$ computable in~time $k^{O(1)}$ such that, if, for any~$e$,
		$g_1(e)$ has $r$-rigidity~$s$ and $g_2(e)$ has rank at~most~$q$,
        then, \MAX{}-3-\SAT{} for formulas with $n$~variables can be~solved in~co-nondeterministic time
        \[
			O(k^2 s + q r^2 k^{\omega - 1}).
        \]
    \end{theorem}

    \begin{proof}
    	Take a~$3$-CNF formula over $n$~variables and an~integer~$t$ and transform~it, using \Cref{theorem:max_3_sat_4_clique}, into a~$4$-partite $3$-uniform hypergraph with parts $H_0, H_1, H_2, H_3$
    	of~size $k = n^{O(1)} 2^{\frac{n}{4}}$. Recall that $G$~contains a~$4$-clique if~and only~if one can satisfy $t$~clauses of~$F$.
        Let~$\mathcal A_0, \mathcal A_1, \mathcal A_2, \mathcal A_3 \in \{0, 1\}^{k \times k \times k}$ be three-dimensional tensors encoding the edges of~$G$. $\mathcal A_i$ is~responsible for storing edges spanning nodes from all parts except for~$H_i$: for example,
        $G$~has an~edge~$(u_1, u_2, u_3)$,
        where $u_1 \in H_1$, $u_2 \in H_2$, and $u_3 \in H_3$,
		if and only if~$\mathcal A_0[u_1, u_2, u_3] = 1$.
		Let
        \[
        R=\sum_{j_0, j_1, j_2, j_3 \in [k]} \mathcal  A_{0}[j_1,j_2,j_3] \cdot \mathcal A_{1}[j_2,j_3,j_0] \cdot \mathcal A_{2}[j_3,j_0,j_1] \cdot \mathcal A_{3}[j_0,j_1,j_2].
        \]
        Then, $G$~has a~$4$-clique if~and only~if $R>0$.
		Further, let
        \[
			R_{j_0} = \sum_{j_1, j_2, j_3 \in [k]} \mathcal A_{0}[j_1, j_2, j_3] \cdot \mathcal A_{1}[j_2, j_3, j_0] \cdot \mathcal A_{2}[j_3, j_0, j_1] \cdot \mathcal A_{3}[j_0, j_1, j_2].
        \]
        Then, $R = R_1 + \dotsb + R_k$.
        For fixed~$j_0$, let $M, L, T \in \{0, 1\}^{k \times k}$ be defined~by
        \begin{align*}
            M[j_2, j_3] &= \mathcal A_1[j_2, j_3, j_0], \\
            L[j_3, j_1] &= \mathcal A_2[j_3, j_0, j_1], \\
            T[j_1, j_2] &= \mathcal A_3[j_0, j_1, j_2].
        \end{align*}
        Thus,
        \[
			R_{j_0} = \sum_{j_1, j_2, j_3 \in [k]} \mathcal A_{0}[j_1, j_2, j_3] M[j_2, j_3] L[j_3, j_1] T[j_1, j_2].
        \]
        Note that $\mathcal A_0$ is a~three-dimensional tensor, whereas $M$,~$L$, and~$T$ are matrices.
        Assume that $\rank(\mathcal A_0) \le q$ and that each of $M$,~$L$, and~$T$ has $r$-rigidity~$s$, that~is,
        \begin{align*}
			\mathcal A_0[j_1, j_2, j_3] &= \sum_{i \in [q]} a_{i, 0}[j_1] b_{i, 0}[j_2] c_{i, 0}[j_3], \\
			M &= R_M + S_M = \left(\sum_{m \in [r]} a_{m, M} \cdot b^T_{m, M}\right) + S_M, \\
			L &= R_L + S_L = \left(\sum_{l \in [r]} a_{l, L} \cdot b^T_{l, L}\right) + S_L, \\
			T &= R_T + S_T = \left(\sum_{t \in [r]} a_{t, T} \cdot b^T_{t, T}\right) + S_T,
        \end{align*}
        where $R_M, S_M, R_L, S_L, R_T, S_T \in \mathbb{F}^{k \times k}$,
        $|S_M| \le s$,
        $|S_L|\le s$, $|S_T| \leq s$, and $\rank(R_M) \le r$, $\rank(R_L) \le r$, $\rank(R_T) \le  r$.

        We~guess such a~representation of $\mathcal A_0$, $M$, $L$, and $T$ (guessing only nonzero entries in~$S_M$, $S_L$, and~$S_T$).
        For fixed $j_0$, one can verify the representations of~$M$,~$L$, and~$T$
		in~time $O(k^2 r^{\omega - 2} + s)$, as the~sum $\sum_{m \in [r]} a_{m, M} \cdot b_{m, M}^{T}$ can~be expressed~as a~product of a~$k \times r$ matrix and a~$r \times k$ matrix.
		Then, \Cref{lemma:matrix_multiplication} yields the~desired running time (the same argument applies to matrices $L$ and~$T$).
		For all $j_0$, this gives a total time complexity of $O(k^3 r^{\omega - 2} + ks)$.
        Since $\mathcal A_0$ does not depend on~$j_0$, it~suffices to~guess and verify its decomposition just once.
        Verifying the decomposition of $\mathcal A_0$ can be done in~time $O(q k^{\omega})$ using~\Cref{lemma:evaluate_sum}.

        Therefore,
        \begin{align*}
			R_{j_0} &= \sum_{j_1, j_2, j_3 \in [k]} \mathcal A_0[j_1, j_2, j_3] \cdot R_M[j_2, j_3] \cdot R_L[j_3, j_1] \cdot R_T[j_1, j_2] +\\
					&+ \underbrace{\sum_{j_1, j_2, j_3 \in [k]} S_M[j_2, j_3] \cdot f_M(j_1, j_2, j_3) + S_L[j_3, j_1] \cdot f_L(j_1, j_2, j_3) + S_T[j_1, j_2] \cdot f_T(j_1, j_2, j_3)}_{R_{j_0}'},
        \end{align*}
        for some functions $f_M$, $f_L$, and $f_T$, which are linear combinations of elements of $\mathcal A_0[j_1, j_2, j_3]$, $R_M[j_2, j_3]$, $S_M[j_2, j_3]$, $R_L[j_3, j_1]$, $S_L[j_3, j_1]$, $R_T[j_1, j_2]$, and $S_T[j_1, j_2]$, and each of which can be computed in~$O(1)$ time at any specific point.
        Since $|S_M|, |S_L|, |S_T| \leq s$, we can compute $R_{j_0}'$ in~time $O(ks)$, for each~$j_0$, and in~time $O(k^2 s)$ overall, as each of $S_M$, $S_L$, and $S_T$ has at most~$s$ nonzero elements.

        To compute $R_{j_0} - R_{j_0}'$, it~suffices to~compute
        \begin{align*}
			&\sum_{j_1, j_2, j_3 \in [k]} \mathcal A_0[j_1, j_2, j_3] \cdot R_M[j_2, j_3] \cdot R_L[j_3, j_1] \cdot R_T[j_1, j_2] =\\
			& \sum_{j_1, j_2, j_3 \in [k]} \left(\sum_{i \in [q]} a_{i, 0}[j_1] \cdot b_{i, 0}[j_2] \cdot c_{i, 0}[j_3]\right) \cdot \left(\sum_{m \in [r]} a_{m, M}[j_2] \cdot b_{m, M}[j_3]\right) \cdot \\
			&\hspace{3.1em} \cdot \left(\sum_{l \in [r]} a_{l, L}[j_3] \cdot b_{l, L}[j_1]\right) \cdot \left(\sum_{t \in [r]} a_{t, T}[j_1] \cdot b_{t, T}[j_2]\right) =\\
			& \sum_{m, l, t \in [r]} \sum_{i \in [q]} \underbrace{\left( \sum_{j_1\in [k]} a_{i, 0}[j_1] \cdot b_{l, L}[j_1] \cdot a_{t, T}[j_1] \right)}_{h_1(i, l, t)} \cdot \underbrace{\left( \sum_{j_2\in [k]} b_{i, 0}[j_2] \cdot a_{m, M}[j_2] \cdot b_{t, T}[j_2] \right)}_{h_2(i, t, m)} \cdot \\
			& \hspace{4.7em} \cdot \underbrace{\left( \sum_{j_3\in [k]} c_{i, 0}[j_3] \cdot b_{m, M}[j_3] \cdot a_{l, L}[j_3] \right)}_{h_3(i, m, l)} =\\
			& \sum_{m, l, t \in [r]} \sum_{i \in [q]} h_1(i, l, t) \cdot h_2(i, t, m) \cdot h_3(i, m, l),
        \end{align*}
		where $h_1, h_2, h_3  \colon [q] \times [r] \times [r] \to \mathbb{F}$ are defined as follows:
        \begin{align*}
			h_1(i, l, t) &= \sum_{j_1\in [k]} a_{i, 0}[j_1] \cdot b_{l, L}[j_1] \cdot a_{t, T}[j_1], \\
			h_2(i, t, m) &= \sum_{j_2\in [k]} b_{i, 0}[j_2] \cdot a_{m, M}[j_2] \cdot b_{t, T}[j_2], \\
			h_3(i, m, l) &= \sum_{j_3\in [k]} c_{i, 0}[j_3] \cdot b_{m, M}[j_3] \cdot a_{l, L}[j_3].
        \end{align*}
		We~compute $h_1$, $h_2$, and $h_3$ for all inputs simultaneously, then evaluate the sum using these precomputed values.
		Assuming the values of $h_1$, $h_2$, and $h_3$ are computed already, the sum can be computed as follows.
		\begin{align*}
			&\sum_{i \in  [r]} \sum_{m, l, t \in [r]} h_1(i, l, t) \cdot h_2(i, t, m) \cdot h_3(i, m, l) = \\
			& \sum_{i \in  [r]} \sum_{l, m \in  [r]} h_3(i, m, l) \cdot \left(\sum_{t \in [r]} h_1(i, l, t) \cdot h_2(i, t, m)\right).
		\end{align*}
		Let $V_{i}, Z_{i} \in \mathbb{F}^{r \times r}$ be two matrices, such that $V_{i}[l, t] = h_1(l, t)$ and $Z_{i}[t, m] = h_2(i, t, m)$, hence we need to evaluate:
		\begin{align*}
			& \sum_{i \in  [r]} \sum_{l, m \in  [r]} h_3(i, m, l) \cdot \left(\sum_{t \in [r]} V_{i}[l, t] \cdot Z_{i}[t, m]\right) = \\
			& \sum_{i \in  [r]} \sum_{l, m \in  [r]} h_3(i, m, l) \cdot (V_{i} \cdot Z_{i})[l, m].\\
		\end{align*}
		Hence, we can calculate $V_{i} \cdot Z_{i}$ in time $O(r^{\omega} q)$ for all $i \in  [q]$ and then evaluate the sum in time $O(r^2 q)$ for fixed $j_0$ and $O(k r^{\omega} q)$ time overall.

		Below, we~show how to~compute~$h_1$ for all inputs.
		For $h_2$~and~$h_3$, this can be~done in~the same fashion.
        Let $P \in \mathbb F^{q \times k}$ and $W \in \mathbb F^{(r \times r) \times k}$ be~defined by $P[i, j_1] = a_{i, 0}[j_1]$ and $W[(l, t), j_1] = b_{l, L}[j_1] \cdot a_{t, T}[j_1]$.
        Then,
        \[
			h_1(i, l, t) = \sum_{j_1 \in [k]} P[i, j_1] \cdot W[(l, t), j_1] = (P \cdot W^T)[i, (l, t)].
        \]
		\Cref{lemma:matrix_multiplication} guarantees that this can be~computed in~time $O(q r^2 k^{\omega - 2})$, for each $j_0$, and in~overall time $O(q r^2 k^{\omega - 1})$.
		Thus, the overall running time of~the described algorithm is
        \[
			O(k^3r^{\omega - 2} + q k^{\omega} + k^2 s + k r^{\omega} q + q r^2 k^{\omega - 1})=O(k^2 s + q r^2 k^{\omega - 1}),
        \] since $k r^{\omega} q = O(q r^2 k^{\omega - 1})$, $k^{3} r^{\omega - 2} = O(q k^{\omega})$ and $q k^{\omega} = O(q r^2 k^{\omega - 1})$.
    \end{proof}

	The following theorem provides a~trade-off between matrix rigidity and tensor rank.

	\begin{theorem} \label{thm:high_rank_tensors}
        Let $\alpha \in [0.5, 1]$ and~$\beta \in [1, 2]$ be~constants satisfying $\alpha \leq \frac{5 - \beta - \omega}{2}$.
		If \MAX{}-3-\SAT{} cannot be~solved in co-nondeterministic time $O(2^{(1 - \varepsilon)n})$, for any $\varepsilon > 0$, then, for all $\delta > 0$, there exist functions $g_1  \colon \{0, 1\}^{\log^{O(1)} k} \to  \mathbb{F}^{k \times k}$ and $g_2  \colon \{0, 1\}^{\log^{O(1)} k} \to \mathbb{F}^{k \times k \times k}$ computable in~time $k^{O(1)}$, such that, for infinitely many~$k$, at~least one of the following is~satisfied, for at~least one~$e$:
		\begin{itemize}
			\item $g_1(e)$ has~$k^{\alpha - \delta}$-rigidity~$k^{2 - \delta}$;
			\item $\rank(g_2(e))$ is at least~$k^{\beta - \delta}$.
		\end{itemize}
	\end{theorem}
	\begin{proof}
		Using~\Cref{thm:high_rigidity}, one can construct a~generator~$g_1$ that outputs~$k \times k$ matrices with~$k^{\frac{1}{2} - \delta}$-rigidity~$k^{2 - \delta}$ for~$k = n^{O(1)} 2^{n / 4}$.
		Hence, from~now on, we~are interested in~rigidity at~least~$k^{0.5}$.

        Let $g_1$ be the function that outputs matrices~$M$,~$L$, and~$T$ from~\Cref{thm:high_rank}.
		Thus,~$g_1$ takes a~3-CNF formula,~$t$,~$j_0$, and a~number from~$\{0, 1, 2\}$ as input.
		The function~$g_2$ outputs a~three-dimensional tensor~$\mathcal{A}_0$, where~the~input to~$g_2$ is a~3-CNF formula and~$t$.
		If~the input to~the generators is not valid, they output empty families.
		Both these generators have a~seed of~size~$n^{O(1)}$, which corresponds to~$\log^{O(1)} k$ and work in~time polynomial in~$k$.
		Assuming that all matrices~$g_1(s)$ have~$r$-rigidity~$s$ and all tensors~$g_2(s)$ have rank at~most~$q$, implies, using \Cref{thm:high_rank}, that~\MAX{}-3-\SAT{} can be~solved in~co-nondeterministic time $O(k^2 s + q r^2 k^{\omega - 1})$.
		Let $s = k^{2 - \delta}$,
		$r = k^{\alpha - \delta}$,
        and~$q = k^{\beta - \delta}$.
        Then, the inequality $\alpha \leq \frac{5 - \beta - \omega}{2}$ implies that the resulting algorithm solves \MAX{}-3-\SAT{}
        in~co-nondeterministic time~$O(k^{4-\delta}) = O(2^{(1 - \delta / 4)n})$.
	\end{proof}

	One way to balance matrix rigidity and tensor rank is to~ensure that a~tensor has superlinear rank while maximizing matrix rigidity, as~demonstrated in~the next corollary.
	\begin{corollary} \label{crly:high_rank_tensors}
			If \MAX{}-3-\SAT{} cannot be~solved in co-nondeterministic time $O(2^{(1 - \varepsilon)n})$ for any $\varepsilon > 0$, then for all $\delta > 0$ there are two polynomial time generators $g_1  \colon \{0, 1\}^{\log^{O(1)} k} \to  \mathbb{F}^{k \times k}$ and $g_2  \colon \{0, 1\}^{\log^{O(1)} k} \to \mathbb{F}^{k \times k \times k}$ such that for infinitely many $k$ at least one of the following is satisfied:
	       \begin{itemize}
	           \item $g_1(s)$ has $k^{2 - \frac{\omega}{2} - \delta}$-rigidity~$k^{2 - \delta}$ , for some $s$.
	           \item $\rank(g_2(s))$ is at least $k^{1 + \delta}$, for some $s$.
	       \end{itemize}
	\end{corollary}

	One can also improve matrix rigidity in~our trade-off by~conditioning on~the value of~$\omega$.
	\thmHighRankTensorsOmega*
	\begin{proof}
		We modify the~second generator in~\Cref{crly:high_rank_tensors} by adding a~new input on which the generator will output a~tensor of matrix multiplication $\mathcal{A}_{\sqrt{k}}$ of size $k \times k \times k$.
		If $\omega = 2$, then the~statement follows from~\Cref{crly:high_rank_tensors}.
		If $\omega \geq 2 + 2 \Delta$ for some $\Delta > 0$, then for infinitely many~$k$, we have $\rank(\mathcal{A}_{\sqrt{k}}) \geq k^{1 + \delta}$.
	\end{proof}

	If~one wants to obtain improved matrix rigidity under weaker assumptions, then the following corollary holds, similar to \Cref{corollary:weaken_rigidity}.

	\begin{corollary}
		If \MAX{}-3-\SAT{} cannot be~solved in co-nondeterministic time $O(2^{0.85 n})$, then, for some $\Delta > 0$, there exist two generators $g_1 \colon \{0, 1\}^{\log^{O(1)} k} \to \mathbb{F}^{k \times k}$ and $g_2 \colon \{0, 1\}^{\log^{O(1)} k} \to \mathbb{F}^{k \times k \times k}$, computable in time polynomial in~$k$, such that, for infinitely many~$k$, at~least one of~the following conditions holds:

		\begin{itemize}
			\item $g_1(s)$ has $k^{0.7}$-rigidity~$k^{1.4}$, for some~$s$;
			\item $\rank(g_2(s))$ is at~least $k^{1 + \Delta}$, for some~$s$.
		\end{itemize}
	\end{corollary}

	We are now ready to~prove our conditional answer to~the open problem from~\cite{DBLP:journals/cc/GoldreichT18}.
	\crlCanonical*
	\begin{proof}
        It~suffices to~verify that $\alpha = \frac{2}{3}$ and $\beta = 1.25$
        satisfy the inequalities in~the statement of~\Cref{thm:high_rank_tensors} for any possible value of $\omega < 2.38$.
		Then,~\cite[Theorem~4.4]{DBLP:series/lncs/0001W20} provides the desired bound on the canonical circuit size.
	\end{proof}

	\section*{Acknowledgments}
	We~are grateful to~the anonymous reviewers for their many helpful comments and additional references, which helped~us not only improve the writing and correct errors, but also prove stronger versions 
	of~some of~our results.

    \bibliographystyle{alpha}
    \bibliography{references}
\end{document}